\documentclass[fullpage,11pt]{article}

\usepackage[sc]{mathpazo}
\usepackage{fullpage,amsthm}
\usepackage{graphicx} 
\usepackage{array} 
\usepackage{amsmath, amssymb, amsfonts, verbatim}
\usepackage{hyphenat,epsfig,subfigure,multirow}
\usepackage[usenames,dvipsnames]{xcolor}

\usepackage{tcolorbox}
\usepackage{balance}

\definecolor{DarkRed}{rgb}{0.5,0.1,0.1}
\definecolor{DarkBlue}{rgb}{0.1,0.1,0.5}

\usepackage{hyperref}
\hypersetup{
colorlinks=true,
pdfnewwindow=true,
citecolor=ForestGreen,
linkcolor=DarkRed,
filecolor=DarkRed,
urlcolor=DarkBlue
}

\usepackage{bm}
\usepackage{url}
\usepackage{xspace} 
\usepackage[mathscr]{euscript}
\usepackage{algorithm}
\usepackage[noend]{algpseudocode}
\makeatletter
\def\BState{\State\hskip-\ALG@thistlm}
\makeatother

\usepackage{cite}
\usepackage{enumerate}

\usepackage[a4paper,margin=1in]{geometry}

\newtheorem{theorem}{Theorem}
\newtheorem{lemma}{Lemma}[section]
\newtheorem{proposition}[lemma]{Proposition}

\newtheorem{claim}[lemma]{Claim}

\newtheorem{definition}{Definition}
\newtheorem{problem}{Problem}
\newtheorem{remark}[lemma]{Remark}
\newtheorem*{claim*}{Claim}

\allowdisplaybreaks

\renewcommand{\qed}{\nobreak \ifvmode \relax \else
      \ifdim\lastskip<1.5em \hskip-\lastskip
      \hskip1.5em plus0em minus0.5em \fi \nobreak
      \vrule height0.75em width0.5em depth0.25em\fi}

\usepackage[T1]{fontenc}
\usepackage[utf8]{inputenc}

\newcommand{\ourinfo}[1]{Department of Computer and Information Science, University of Pennsylvania. Supported in part by National Science
  Foundation grants CCF-1116961, CCF-1552909, and IIS-1447470 and an Adobe research award.
\newline\noindent Email: \texttt{#1}.}

\newcommand{\Index}{\textnormal{\textsf{Index}}\xspace}
\newcommand{\Indexnk}{\textnormal{\ensuremath{\textsf{Index}^n_k}}\xspace}
\newcommand{\Trap}{\textnormal{\textsf{Trap}}\xspace}

\newcommand{\distI}{\ensuremath{\dist_{\Index}}}
\newcommand{\distT}{\ensuremath{\dist_{\trap}}}

\newcommand{\distdet}{\textnormal{\ensuremath{\dist_{\textsf{det}}}}}
\newcommand{\distsi}{\textnormal{\ensuremath{\dist_{\textsf{SI}}}}}

\newcommand{\bS}{\ensuremath{\bm{S}}}

\newcommand{\bP}{\ensuremath{\bm{P}}}
\newcommand{\bPP}{\ensuremath{\bm{\mathcal{P}}}}
\newcommand{\bI}{\ensuremath{\bm{I}}}
\newcommand{\bT}{\ensuremath{\bm{\theta}}}

\newcommand{\bX}{\ensuremath{\bm{X}}}

\newcommand{\bA}{\ensuremath{\bm{A}}}
\newcommand{\bB}{\ensuremath{\bm{B}}}
\newcommand{\bC}{\ensuremath{\bm{C}}}
\newcommand{\bD}{\ensuremath{\bm{D}}}

\newcommand{\supp}[1]{\ensuremath{\textsc{supp}(#1)}}

\newcommand{\prot}{\ensuremath{\Pi}}
\newcommand{\bM}{\ensuremath{\bm{M}}}
\renewcommand{\SS}{\ensuremath{\mathcal{S}}}

\newcommand{\alg}{\ensuremath{\mathcal{A}}\xspace}

\newcommand{\prottrap}{\ensuremath{\Pi}_{\trap}}
\newcommand{\protindex}{\ensuremath{\Pi}_{\Index}}
\newcommand{\protext}{\ensuremath{\Pi}_{\textsf{Ext}}}
\newcommand{\protsc}{\ensuremath{\Pi}_{\textsf{SC}}}
\newcommand{\protsi}{\ensuremath{\Pi}_{\textsf{SI}}}

\newcommand{\bSS}{\ensuremath{\bm{\SS}}}
\newcommand{\algline}{
  \rule{0.5\linewidth}{.1pt}\hspace{\fill}%
  \par\nointerlineskip \vspace{.1pt}
}

\newcommand{\bE}{\ensuremath{\bm{E}}}
\newcommand{\bR}{\ensuremath{\bm{R}}}

\newcommand{\trap}{\textnormal{\textsf{Trap}}\xspace}
\newcommand{\errs}{\textnormal{~errs}}
\newcommand{\Yes}{\textnormal{\textsf{YES}}\xspace}
\newcommand{\No}{\textnormal{\textsf{NO}}\xspace}
\newcommand{\distIY}{\distI^{\textsf{Y}}}
\newcommand{\distIN}{\distI^{\textsf{N}}}

\newcommand{\Paren}[1]{\Big(#1\Big)}
\newcommand{\Bracket}[1]{\Big[#1\Big]}

\newcommand{\ICost}[2]{\ensuremath{\textnormal{\textsf{ICost}}_{#2}(#1)}\xspace}
\newcommand{\IC}[3]{\ensuremath{\textnormal{\textsf{IC}}_{#2}^{#3}(#1)}\xspace}
\newcommand{\CC}[3]{\ensuremath{\textnormal{\textsf{CC}}_{#2}^{#3}(#1)}\xspace}

\newcommand{\SSt}{\ensuremath{\mathcal{S}}}

\newcommand{\tester}{\textnormal{\textit{Tester}}\xspace}

\newcommand{\accept}{\textnormal{\textsf{ACCEPT}}\xspace}
\newcommand{\reject}{\textnormal{\textsf{REJECT}}\xspace}

\newcommand{\bprotindex}{\ensuremath{\bm{\protindex}}}
\newcommand{\bprottrap}{\ensuremath{\bm{\prottrap}}}

\newcommand{\bprotsc}{\ensuremath{\bm{\protsc}}}
\newcommand{\bprot}{\ensuremath{\bm{\prot}}}

\newcommand{\dista}{\ensuremath{\dist_{\textsf{apx}}}}
\newcommand{\diste}{\ensuremath{\dist_{\textsf{est}}}}
\newcommand{\dister}{\ensuremath{\dist_{\textsf{ext}}}}
\newcommand{\distnew}{\ensuremath{\dist_{\textsf{new}}}}

\newenvironment{tbox}{\begin{tcolorbox}[
		enlarge top by=5pt,
		enlarge bottom by=5pt,
		 boxsep=0pt,
                  left=4pt,
                  right=4pt,
                  top=10pt,
                  arc=0pt,
                  boxrule=1pt,toprule=1pt,
                  colback=white
                  ]
	}
{\end{tcolorbox}}

\renewenvironment{proof}[1][Proof]{\paragraph{#1.}\xspace}{\hfill$\qed$}

\newcommand{\ssection}[1]{\noindent\textbf{#1}}
\newcommand{\toShrink}{-.25cm}
\newcommand{\toShrinkEnu}{-.2cm}
\newcommand{\toShrinkEqn}{-.45cm}

\newcommand{\toShrinkTextbox}{0cm}
\newcommand{\cTextbox}[2]{\vspace{\toShrinkTextbox}\textbox{#1}{\vspace{\toShrinkTextbox} #2 \vspace{\toShrinkTextbox}}}

\renewcommand{\bar}[1]{\overline{#1}}

\newcommand{\etal}{{\it et al.\,}}
\newcommand{\eps}{\epsilon}

\newcommand{\event}{\mathcal{E}}

\newcommand{\card}[1]{\left\vert{#1}\right\vert}

\newcommand{\opt}{\ensuremath{{opt}}\xspace}

\newcommand{\ceil}[1]{{\left\lceil{#1}\right\rceil}}

\newcommand{\set}[1]{\ensuremath{\left\{ #1 \right\}}}
\newcommand{\seq}[1]{{\left< #1 \right>}}

\newcommand{\poly}{\mbox{\rm poly}}

\newcommand{\REM}[1]{}
\newcommand{\Ot}{\widetilde{O}}
\newcommand{\Omgt}{\widetilde{\Omega}}

\newcommand{\union}{\ensuremath{\cup}}
\newcommand{\pair}[2]{\langle {#1},{#2} \rangle}

\newcommand{\textbox}[2]{
{
\begin{tbox}
\textbf{#1}
{#2}
\end{tbox}
}
}

\newcommand{\eat}[1]{}

\DeclareMathOperator*{\Exp}{\ensuremath{\textnormal{E}}}
\DeclareMathOperator*{\Prob}{\ensuremath{\textnormal{Pr}}}
\renewcommand{\Pr}{\Prob}
\newcommand{\Ex}{\Exp}

\newcommand{\norm}[1]{\ensuremath{\|#1\|}}

\newcommand{\PR}[1]{\ensuremath{\Pr\left(#1\right)}\xspace}

\newcommand{\paren}[1]{\ensuremath{\left(#1\right)}\xspace}

\newcommand{\dist}{\ensuremath{\mathcal{D}}}

\newcommand{\SA}{\ensuremath{\SS}}
\newcommand{\Salpha}{\ensuremath{\widehat{\mathcal{S}}}}

\newcommand{\FC}{\ensuremath{\mathcal{F}}}

\newcommand{\setCoverEst}{\textnormal{\ensuremath{\textsf{SetCover}_{\textsf{est}}}}\xspace}
\newcommand{\setCoverApx}{\textnormal{\ensuremath{\textsf{SetCover}_{\textsf{apx}}}}\xspace}

\newcommand{\sparseIndexingnk}{\textnormal{\ensuremath{\textsf{SparseIndexing}^N_k}}\xspace}
\newcommand{\istar}{\ensuremath{{i^*}}\xspace}
\newcommand{\estar}{\ensuremath{{e^*}}\xspace}
\newcommand{\SAistar}{\ensuremath{{\SA^{-\istar}}}\xspace}

\newcommand{\hevent}{\widehat{\event}}

\newcommand{\CILP}{\ensuremath{\textsf{ILP}_{\textsf{Cover}}}}
\newcommand{\amax}{\ensuremath{a_{\textsf{max}}}}
\newcommand{\bmax}{\ensuremath{b_{\textsf{max}}}}
\newcommand{\cmax}{\ensuremath{c_{\textsf{max}}}}

\newcommand{\aT}{\ensuremath{\widetilde{a}}}
\newcommand{\cT}{\ensuremath{\widetilde{c}}}

\newcommand{\lOne}[1]{\ensuremath{\left\| #1\right\|_1}}

\newcommand{\Ins}{\ensuremath{\mathcal{I}}}

\newcommand{\RIns}{\ensuremath{\mathcal{I}_{\textsf{R}}}}

\newcommand{\testerILP}{\textnormal{\ensuremath{\tester}}}

\newcommand{\cost}{\textnormal{\textit{Cost}}}

\newcommand{\bres}{\textnormal{\ensuremath{b_{\textsf{res}}}}}

\newcommand{\AT}{\widetilde{A}}
\newcommand{\AH}{\ensuremath{\widehat{A}}}

\newcommand{\rou}{\rho} 

\title{Tight Bounds for Single-Pass Streaming Complexity of the Set Cover Problem}
\author{Sepehr Assadi\thanks{\ourinfo{\{sassadi,sanjeev,yangli2\}@cis.upenn.edu}}\and
Sanjeev Khanna\footnotemark[1] \and
Yang Li\footnotemark[1]  
}

\date{}

\begin{document}
\maketitle

\thispagestyle{empty}
\begin{abstract}
  We resolve the space complexity of \emph{single-pass} streaming algorithms for
  approximating the classic set cover problem. For finding an $\alpha$-approximate set
  cover (for any $\alpha= o(\sqrt{n})$) using a single-pass streaming algorithm, we show that
  $\Theta(mn/\alpha)$ space is both sufficient and necessary (up to an $O(\log{n})$ factor); here $m$
  denotes number of the sets and $n$ denotes size of the universe.  This
  provides a strong negative answer to the open question posed by
  Indyk~\etal~\cite{IndykMV15} regarding the possibility of having a single-pass algorithm
  with a small approximation factor that uses sub-linear space.
	
  We further study the problem of \emph{estimating} the \emph{size} of a minimum set cover
  (as opposed to finding the actual sets), and establish that an additional factor of
  $\alpha$ saving in the space is achievable in this case and that 
 this is the best possible. In other words, we show that
  $\Theta(mn/\alpha^2)$ space is both sufficient and necessary (up to logarithmic factors)
  for estimating the size of a minimum set cover to within a factor of $\alpha$. Our
algorithm in fact works for the more general problem of estimating the optimal value of a \emph{covering integer program}.
 On the other hand, our lower bound holds even for set cover instances where the sets are presented in a \emph{random order}.
\end{abstract}

\clearpage
\setcounter{page}{1}

\section{Introduction}\label{sec:intro}

The \emph{set cover} problem is a fundamental optimization problem with many applications
in computer science and related disciplines. The input is a
universe $[n]$ and a collection of $m$ subsets of $[n]$, $\SSt = \seq{S_1,\ldots,S_m}$,
and the goal is to find a subset of $\SSt$ with the \emph{smallest} cardinality that
\emph{covers} $[n]$, i.e., whose union is $[n]$; we call such a collection of sets a
\emph{minimum set cover} and throughout the paper denote its cardinality by $\opt :=
\opt(\SSt)$.

The set cover problem can be formulated in the well-established streaming
model~\cite{AlonMS96,Muth05}, whereby the sets in $\SS$ are presented one by one in a stream and the goal is to solve the set cover
problem using a \emph{space-efficient} algorithm. The streaming setting for the set
cover problem has been studied in several recent work, including~\cite{SahaG09,EmekR14,DemaineIMV14,IndykMV15,ChakrabartiW15}.
We refer the interested reader to these
references for many applications of the set cover problem in the streaming model.  In this
paper, we focus on algorithms that make only \emph{one pass} over the stream (i.e.,
single-pass streaming algorithms), and our goal is to settle the space complexity of
single-pass streaming algorithms that \emph{approximate} the set cover problem.

Two versions of the set cover problem are considered in this paper: $(i)$ computing a
minimum set cover, and $(ii)$ computing the size of a minimum set cover.
Formally,

\begin{definition}[$\alpha$-approximation]
  An algorithm $\alg$ is said to \emph{$\alpha$-approximate} the set cover problem iff on
  every input instance $\SS$, $\alg$ outputs a collection of (the indices of) at most
  $\alpha\cdot \opt$ sets that covers $[n]$, along with a \emph{certificate of covering}
  which, for each element $e \in [n]$, specifies the set used for covering $e$.  If $\alg$
  is a randomized algorithm, we require that the certificate corresponds to a valid set
  cover w.p.\footnote{Throughout, we use \emph{w.p.} and \emph{w.h.p.} to abbreviate
    ``with probability'' and ``with high probability'', respectively.}  at least $2/3$.
\end{definition}

We remark that the requirement of returning a
certificate of covering is standard in the literature (see, e.g.,~\cite{EmekR14,ChakrabartiW15}).

\begin{definition}[$\alpha$-estimation]
  An algorithm $\alg$ is said to \emph{$\alpha$-estimate} the set cover problem iff on
  every input instance $\SS$, $\alg$ outputs an estimate for the cardinality of a minimum
  set cover in the range $[\opt, \alpha \cdot \opt]$.  If $\alg$ is a randomized
  algorithm, we require that:
  \begin{align*}
	\Pr\Paren{\alg(\SS) \in [\opt, \alpha \cdot \opt]} \geq 2/3
  \end{align*}
\end{definition} 

\subsection{Our Results}\label{sec:results}

We resolve the space complexities of both versions of the set cover problem. Specifically,
we show that for any $\alpha = o(\sqrt{n}/\log{n})$ and any $m = \poly(n)$,

\begin{itemize}
\item There is a \emph{deterministic} single-pass streaming algorithm that
  \emph{$\alpha$-approximates} the set cover problem using space
  $\Ot(mn/\alpha)$ bits and moreover, any single-pass streaming algorithm (possibly
  \emph{randomized}) that $\alpha$-approximates the set cover problem must use
  $\Omega(mn/\alpha)$ bits of space.
\item There is a \emph{randomized} single-pass streaming algorithm that
  \emph{$\alpha$-estimates} the set cover problem using space $\Ot(mn/\alpha^2)$
  bits and moreover, any
  single-pass streaming algorithm (possibly \emph{randomized}) that $\alpha$-estimates the
  set cover problem must use $\Omgt(mn/\alpha^2)$ bits of space.
\end{itemize}

We should point out right away that in this paper, we are \emph{not} concerned with
poly-time computability, though, our algorithms for set cover can be made
computationally efficient for any $\alpha \geq \log{n}$ by allowing an extra $\log{n}$
factor in the space requirement\footnote{Set cover admits a classic $\log{n}$-approximation algorithm~\cite{Johnson74a,Slavik97}, and unless \textsf{P = NP}, there is no polynomial time
$\alpha$-approximation for the set cover problem for $\alpha < (1-\eps)\log{n}$ (for any
constant $\eps > 0$)~\cite{DinurS14,Feige98,LundY94}.}.

We establish our upper bound result for $\alpha$-estimation for a much more general
problem: estimating the optimal value of a \emph{covering integer linear program} (see
Section~\ref{sec:general} for a formal definition).  Moreover, the space lower bound for
$\alpha$-estimation (for the original set cover problem) holds even if the sets are
presented in a \emph{random order}.  We now describe each of these two sets of results in
more details.

\paragraph{Approximating Set Cover.}  There is a very simple deterministic 
$\alpha$-approximation algorithm for the set cover problem using space $\Ot(mn/\alpha)$ bits which we mention in Section~\ref{sec:techniques} for completeness. 
Perhaps surprisingly, we establish that this simple algorithm is essentially the best possible; any $\alpha$-approximation algorithm for the set cover 
problem requires $\Omgt(mn/\alpha)$ bits of space (see Theorem~\ref{thm:find-lower} for a formal statement). 

Prior to our work, the best known lower bounds for single-pass streams ruled out
$(3/2-\eps)$-approximation using $o(mn)$ space~\cite{IndykMV15} (see also~\cite{HarpeledIMV16}),
$o(\sqrt{n})$-approximation in $o(m)$ space~\cite{EmekR14,ChakrabartiW15}, and
$O(1)$-approximation in $o(mn)$ space~\cite{DemaineIMV14} (only for \emph{deterministic}
algorithms); see Section~\ref{sec:related} for more detail on different existing lower
bounds.  Note that these lower bound results leave open the possibility
 of a single-pass randomized $3/2$-approximation or even a deterministic $O(\log{n})$-approximation algorithm for the
 set cover problem using only $\Ot(m)$ space.  Our
result on the other hand, rules out the possibility of having any \emph{non-trivial}
trade-off between the approximation factor and the space requirement, answering an open
question raised by Indyk \etal~\cite{IndykMV15} in the strongest sense possible.

We should also point out that the bound of $\alpha = o(\sqrt{n}/\log{n})$ in our lower
bound is tight up to an $O(\log{n})$ factor since an $O(\sqrt{n})$-approximation is known
to be achievable in $\Ot(n)$ space (essentially independent of $m$ for $m =
\poly(n)$)~\cite{EmekR14,ChakrabartiW15}.

\paragraph{Estimating Set Cover Size.} 
We present an $\Ot(mn/\alpha^2)$ space algorithm for $\alpha$-estimating the set cover
problem, and in general any covering integer program
(see Theorem~\ref{thm:general} for a formal statement). Our upper bound suggests that if
one is only interested in $\alpha$-estimating the size of a minimum set cover (instead of
knowing the actual sets), then an additional $\alpha$ factor saving in the space (compare to the best
possible $\alpha$-approximation algorithm) is possible. To the best of our knowledge,
this is the first non-trivial \emph{gap} between the space complexity of
$\alpha$-approximation and $\alpha$-estimation for the set cover problem.

We further show that the space complexity of our $\Ot(mn/\alpha^2)$ space $\alpha$-estimation algorithm for the set cover problem is essentially
tight (up to logarithmic factors). In other words, any $\alpha$-estimation algorithm for set cover (possibly randomized) requires $\Omgt(mn/\alpha^2)$ space (see Theorem~\ref{thm:rand-lower} for a formal statement).

There are examples of classic optimization problems in the streaming literature for
which size estimation seems to be distinctly easier in the \emph{random arrival
  streams}\footnote{In random arrival streams, the input (in our case, the collection of
  sets) is randomly permuted before being presented to the algorithm} compare to the
\emph{adversarial streams} (see, e.g.,~\cite{KapralovKS14}). However, we show that this is
\emph{not} the case for the set cover problem, i.e., the lower bound of
$\Omgt(mn/\alpha^2)$ for $\alpha$-estimation continues to hold even for random arrival
streams.

We note in passing two other results also: $(i)$ our bounds for $\alpha$-approximation and
$\alpha$-estimation also prove \emph{tight} bounds on the \emph{one-way communication
  complexity} of the \emph{two-player} communication model of the set cover problem (see
Theorem~\ref{thm:two-find-adv} and Theorem~\ref{thm:two-adv}), previously studied
in~\cite{Nisan02,DemaineIMV14,ChakrabartiW15}; $(ii)$ the use of randomization in our
$\alpha$-estimation algorithm is inevitable: any \emph{deterministic} $\alpha$-estimation
algorithm for the set cover requires $\Omega(mn/\alpha)$ bits of space (see
Theorem~\ref{thm:det-lower}).

\subsection{Our Techniques}\label{sec:techniques}

\paragraph{Upper bounds.}  An $\alpha$-approximation using $\Ot(mn/\alpha)$ bits of space
can be simply achieved as follows. {Merge} (i.e., take the union of) every $\alpha$ sets in
the stream into a single set; at the end of the stream, solve the set cover problem over
the merged sets. To recover a certificate of covering, we also record for each element $e$
in each merged set, any set in the merge that covers $e$. It is an easy exercise to
verify that this algorithm indeed achieves an $\alpha$-approximation and can be
implemented in space $\Ot(mn/\alpha)$ bits.

Our $\Ot(mn/\alpha^2)$-space $\alpha$-estimation algorithm is more involved and in fact works for any covering integer program (henceforth, a \emph{covering ILP} for short). 
We follow the line of work in~\cite{DemaineIMV14} and~\cite{IndykMV15} by performing ``dimensionality reduction'' over
the sets (in our case, columns of the constraint matrix $A$) and storing their projection over a randomly sampled subset of the universe (here, constraints) during the
stream. However, the goal of our \emph{constraint sampling} approach is entirely different from the ones in~\cite{DemaineIMV14,IndykMV15}. The \emph{element sampling} approach of~\cite{DemaineIMV14,IndykMV15} aims to find a ``small'' cover of the sampled universe which also covers the vast majority of the elements in the original universe. 
This allows the algorithm to find a small set cover of the sampled universe in one pass
while reducing the number of remaining uncovered elements for the next pass; hence, applying this approach
repeatedly over \emph{multiple} passes on the input allows one to obtain a complete cover. 

On the other hand, the goal of our {constraint sampling} approach is to create a smaller
instance of set cover (in general, covering ILP) with the property that the minimum set
cover \emph{size} of the sampled instance is a ``proxy'' for the minimum set cover size of
the original instance. We crucially use the
fact that the algorithm does \emph{not} need to identify the actual cover and hence it can
estimate the size of the solution based on the optimum set cover size in the sampled
universe.

At the core of our approach is a simple yet very general lemma, referred to as the
\emph{constraint sampling lemma} (Lemma~\ref{lem:sub-sample-save-opt}) which may be of
independent interest.  Informally, this lemma states that for any covering ILP instance
$\Ins$, the optimal value of a sub-sampled instance $\RIns$, obtained by picking roughly
$1/\alpha$ fraction of the constraints uniformly at random from $\Ins$, is an $\alpha$ estimator of the optimum value of $\Ins$ whenever no constraint is ``too hard'' to satisfy.

Nevertheless, the constraint sampling is not enough for reducing the space to meet the desired $\Ot(mn/\alpha^2)$ bound (see Theorem~\ref{thm:general}). 
Hence, we combine it with a \emph{pruning}
step, similar to the ``set filtering'' step of~\cite{IndykMV15} for (unweighted) set cover
(see also ``GreedyPass'' algorithm of~\cite{ChakrabartiW15}) to sparsify the columns in
the input matrix $A$ before performing the sampling.  We point out that as the variables
in $\Ins$ can have different weights in the objective function (e.g. for weighted set
cover), our pruning step needs to be sensitive to the weights.


\paragraph{Lower bounds.}  As is  typical in the streaming literature, our lower
bounds are obtained by establishing \emph{communication complexity} lower bounds; in
particular, in the \emph{one-way two-player} communication model. To prove these bounds,
we use the \emph{information complexity} paradigm, which allows one to reduce the problem,
via a direct sum type argument, to multiple instances of a simpler problem. For our lower
bound for $\alpha$-estimation, this simpler problem turned out to be a variant of the
well-known \emph{Set Disjointness} problem. However, for the lower bound of
$\alpha$-approximation algorithms, we introduce and analyze a new intermediate problem, called the
\emph{Trap} problem.

The Trap problem is a \emph{non-boolean} problem defined as follows: Alice is given a set
$S$, Bob is given a set $E$ such that all elements of $E$ belong to $S$ except for a
\emph{special element} $e^*$, and the goal of the players is to ``trap'' this special
element, i.e., to find a \emph{small} subset of $E$ which contains $e^*$.  For our
purpose, Bob only needs to trap $e^*$ in a set of cardinality $\card{E}/{2}$. To prove a
lower bound for the Trap problem, we design a novel reduction from the well-known
\emph{Index} problem, which requires Alice and Bob to use the protocol for the Trap problem over non-legal inputs (i.e., the ones for which the Trap
problem is not well-defined), while ensuring that they are not being ``fooled'' by
the output of the Trap protocol over these inputs.

To prove our lower bound for $\alpha$-estimation in random arrival streams, we follow the
approach of~\cite{ChakrabartiCM08} in proving the communication complexity lower bounds
when the input data is \emph{randomly allocated} between the players (as opposed to
adversarial partitions). However, the distributions and the problem considered in this
paper are different from the ones in~\cite{ChakrabartiCM08}.

\subsection{Related Work}\label{sec:related}

Communication complexity of the set cover problem was first studied by
Nisan~\cite{Nisan02}.  Among other things, Nisan showed that the two-player communication complexity of $(\frac{1}{2}-\eps)\log{n}$-estimating the set cover
is $\Omega(m)$. In particular, this implies that any constant-pass streaming algorithm
that $(\frac{1}{2}-\eps)\log{n}$-estimates the set cover must use $\Omega(m)$
bits of space.

Saha and Getoor~\cite{SahaG09} initiated the study of set cover in the \emph{semi-streaming} model~\cite{FKMSZ05} where the 
sets are arriving in a stream and the algorithms are required to use $\Ot(n)$ space, and obtained an $O(\log{n})$-approximation via an $O(\log{n})$-pass algorithm that uses $O(n\log{n})$ space. 
A similar performance was also achieved by~\cite{CormodeKW10} in the context of ``disk friendly'' algorithms. As designed, the algorithm of~\cite{CormodeKW10} achieves $(1 + \beta\ln{n})$-approximation by making
$O(\log_\beta{n})$ passes over the stream using $O(n \log{n})$ space.

The \emph{single-pass semi-streaming} setting for set cover was initially and throughly
studied by Emek and Ros{\'{e}}n~\cite{EmekR14}. They provided an $O(\sqrt{n})$-approximation
using $\Ot(n)$ space (which also extends to the weighted set cover problem) and a lower
bound that states that no semi-streaming algorithm (i.e., an algorithm using only $\Ot(n)$ space)
that $O(n^{1/2-\eps})$-estimates set cover exists.  Recently, Chakrabarti and
Wirth~\cite{ChakrabartiW15} generalized the space bounds in~\cite{EmekR14} to
\emph{multi-pass} algorithms, providing an almost complete understanding of the pass/approximation tradeoff for semi-streaming algorithms.
In particular, they developed a deterministic $p$-pass $(p+1)\cdot n^{1/(p+1)}$-approximation
algorithm in $\Ot(n)$ space and prove that any $p$-pass $n^{1/(p+1)}/(c\cdot(p+1)^2)$-estimation algorithm requires $\Omega(n^{c}/p^3)$ space for
some constant $c > 1$ ($m$ in their ``hard instances'' is $\Theta(n^{cp})$). This, in particular,
implies that any \emph{single-pass} $o(\sqrt{n})$-estimation algorithm requires
$\Omega(m)$ space.

Demaine~\etal~\cite{DemaineIMV14} studied the trade-off between the number of passes, the
approximation ratio, and the space requirement of general streaming algorithms (i.e., not
necessarily semi-streaming) for the set cover problem and developed an algorithm that for
any $\delta = \Omega(1/\log{n})$, makes $O(4^{1/\delta})$ passes over the stream and
achieves an $O(4^{1/\delta} \rou)$-approximation using $\Ot(mn^\delta)$ space; here $\rou$
is the approximation factor of the \emph{off-line} algorithm for solving the set cover
problem.  The authors further showed that any \emph{constant-pass} \emph{deterministic}
$O(1)$-estimation algorithm for the set cover requires $\Omega(mn)$ space. Very recently,
Indyk~\etal~\cite{IndykMV15} (see also~\cite{HarpeledIMV16}) made a significant improvement on the trade-off achieved
by~\cite{DemaineIMV14}: they presented an algorithm that for any $\delta > 0$, makes
$O(1/\delta)$ passes over the stream and achieves an $O(\rou/\delta)$-approximation using
$\Ot(mn^\delta)$ space. The authors also established two lower bounds: for multi-pass
algorithms, any algorithm that computes an \emph{optimal} set cover solution while making
only $({1 \over 2\delta}-1)$ passes must use $\Omgt(mn^\delta)$ space. More relevant to
our paper, they also showed that any \emph{single-pass} streaming algorithm (possibly
randomized) that can distinguish between the instances with set cover size of $2$ and $3$
w.h.p., must use $\Omega(mn)$ bits.

\paragraph{Organization.} We introduce in Section~\ref{sec:prelim} some preliminaries needed for the rest of the paper.
In Section~\ref{sec:lower-find}, we present our ${\Omega}(mn/\alpha)$ space lower bound for computing an
$\alpha$-approximate set cover in a single-pass. In Section~\ref{sec:general}, we present
a single-pass streaming algorithm for estimating the optimal value of a covering integer
program and prove an $\Ot(mn/\alpha^2)$ upper bound on the space complexity of $\alpha$-estimating the (weighted) set cover problem. 
In Section~\ref{sec:rand-lower}, we present our ${\Omega}(mn/\alpha^2)$ space
lower bound for $\alpha$-estimating set cover in a single-pass. Finally, Section~\ref{sec:det-lower} contains
our $\Omega(mn/\alpha)$ space lower bound for \emph{deterministic} $\alpha$-estimation algorithms.

\section{Preliminaries}\label{sec:prelim}

\paragraph{Notation.} 
We use bold face letters to represent random variables. For any random variable $\bX$,
$\supp{\bX}$ denotes its support set. We define $\card{\bX}: = \log{\card{\supp{\bX}}}$.  For any $k$-dimensional tuple $X = (X_1,\ldots,X_k)$ and any $i \in [k]$, we
define $X^{<i}:=(X_1,\ldots,X_{i-1})$, and $X^{-i}:=(X_1,\ldots,X_{i-1},X_{i+1},\ldots,X_{k})$.  The notation ``$X\in_R U$'' indicates that $X$ is chosen uniformly at random from a set $U$.
Finally, we use upper case letters (e.g. $M$) to represent matrices and lower case letter (e.g. $v$) to represent vectors.

\paragraph{Concentration Bounds.} We use an extension of the Chernoff-Hoeffding bound for \emph{negatively correlated} random variables. Random variables $\bX_1,\ldots,\bX_n$ are 
negatively correlated if for every set $S \subseteq [n]$, $\Pr\paren{\wedge_{i \in S} \bX_i = 1} \leq \prod_{i \in S}\Pr\paren{\bX_i = 1}$. It was first proved in~\cite{PanconesiS97} that
the Chernoff-Hoeffding bound continues to hold for the case of random variables that satisfy this generalized version of negative correlation (see also~\cite{ImpagliazzoK10}).

\subsection{Tools from Information Theory}\label{sec:info}
We briefly review some basic concepts from information theory needed for establishing our lower bounds. For a
broader introduction to the field, we refer the reader to the excellent text by Cover and
Thomas~\cite{ITbook}.

In the following, we denote the \emph{Shannon Entropy} of a random variable $\bA$ by
$H(\bA)$ and the \emph{mutual information} of two random variables $\bA$ and $\bB$ by
$I(\bA;\bB) = H(\bA) - H(\bA \mid \bB) = H(\bB) - H(\bB \mid \bA)$. If the distribution
$\dist$ of the random variables is not clear from the context, we use $H_\dist(\bA)$
(resp. $I_{\dist}(\bA;\bB)$). We use $H_2$ to denote the binary entropy function where for any real number $0
< \delta < 1$, $H_2(\delta) = \delta\log{\frac{1}{\delta}} +
(1-\delta)\log{\frac{1}{1-\delta}}$.

We use the following basic properties of entropy and mutual information (proofs can be
found in~\cite{ITbook}, Chapter~2).
\begin{claim}\label{clm:it-facts}
  Let $\bA$, $\bB$, and $\bC$ be three random variables.
  \begin{enumerate}
  \item \label{part:uniform} $0 \leq H(\bA) \leq \card{\bA}$. $H(\bA) = \card{\bA}$
    iff $\bA$ is uniformly distributed over its support.
  \item \label{part:info-zero} $I(\bA ; \bB) \geq 0$. The equality holds iff $\bA$ and
    $\bB$ are \emph{independent}.
  \item \label{part:cond-reduce} \emph{Conditioning on a random variable reduces entropy}:
    $H(\bA \mid \bB,\bC) \leq H(\bA \mid \bB)$. The equality holds iff $\bA$ and $\bC$
    are independent conditioned on $\bB$.
  \item \label{part:sub-additivity} \emph{Subadditivity of entropy}: $H(\bA,\bB \mid \bC)
    \leq H(\bA \mid \bC) + H(\bB \mid \bC)$.
  \item \label{part:chain-rule} \emph{The chain rule for mutual information}: $I(\bA,\bB ;
    \bC) = I(\bA ; \bC) + I( \bB;\bC \mid \bA)$.
  \item \label{part:ent-event} For any event $E$ independent of $\bA$ and $\bB$, $H(\bA
    \mid \bB, E) = H(\bA \mid \bB)$.
  \item \label{part:info-event} For any event $E$ independent of $\bA,\bB$ and $\bC$,
    $I(\bA;\bB \mid \bC,E) = I(\bA;\bB \mid \bC)$.
  \end{enumerate}
\end{claim}

The following claim (Fano's inequality) states that if a random variable $\bA$ can be used to estimate the value
of another random variable $\bB$, then $\bA$ should ``consume'' most of the entropy of
$\bB$.

\begin{claim}[Fano's inequality]\label{clm:fano}
For any binary random variable $\bB$ and any (possibly randomized) function f that
predicts $\bB$ based on $\bA$, if $\Pr($f(\bA)$ \neq \bB) = \delta$, then $H(\bB \mid \bA)
\leq H_2(\delta)$.
\end{claim}

We also use the following simple claim, which states that conditioning on independent
random variables can only increase the mutual information.
\begin{claim}\label{clm:info-increase}
  For any random variables $\bA, \bB, \bC$, and $\bD$, if $\bA$ and $\bD$ are independent
  conditioned on $\bC$, then $I(\bA; \bB \mid \bC) \leq I(\bA; \bB \mid \bC, \bD)$.
\end{claim}
\begin{proof}
  Since $\bA$ and $\bD$ are independent conditioned on $\bC$, by
  Claim~\ref{clm:it-facts}-(\ref{part:cond-reduce}), $H(\bA \mid \bC) = H(\bA \mid
  \bC,\bD)$ and $H(\bA \mid \bC,\bB) \ge H(\bA \mid \bC,\bB,\bD)$.  We have,
	\begin{align*}
	  I(\bA; \bB \mid \bC) &= H(\bA \mid \bC) - H(\bA \mid \bC,\bB) = H(\bA \mid \bC,\bD) - H(\bA \mid \bC,\bB) \\
	  &\leq H(\bA \mid \bC,\bD) - H(\bA \mid \bC,\bB,\bD) = I(\bA;\bB \mid \bC,\bD)
	\end{align*}
\end{proof}

\subsection{Communication Complexity and Information Complexity}\label{sec:cc-ic}
Communication complexity and information complexity play an important role in our lower bound proofs. 
We now provide necessary definitions for completeness.

\paragraph{Communication complexity.} Our lowers bounds for single-pass streaming algorithms
are established through communication complexity lower bounds.  Here, we briefly
provide some context necessary for our purpose; for a more detailed treatment of
communication complexity, we refer the reader to the excellent text by Kushilevitz and
Nisan~\cite{KN97}.

We focus on the \emph{two-player one-way communication} model. Let $P$ be a relation with
domain $\mathcal{X} \times \mathcal{Y} \times \mathcal{Z}$.  Alice receives an input $X
\in \mathcal{X}$ and Bob receives $Y \in \mathcal{Y}$, where $(X,Y)$ are chosen from a
joint distribution $\dist$ over $\mathcal{X} \times \mathcal{Y}$.  In addition to private randomness, the players also have an
access to a shared public tape of random bits $R$. Alice sends a single message $M(X,R)$
and Bob needs to output an answer $Z := Z(M(X,R),Y,R)$ such that $(X,Y,Z) \in P$.

We use $\prot$ to denote a protocol used by the players. Unless specified otherwise, we
always assume that the protocol $\prot$ can be randomized (using both public and
private randomness), \emph{even against a prior distribution $\dist$ of inputs}. For any
$0 < \delta < 1$, we say $\prot$ is a $\delta$-error protocol for $P$ over a distribution
$\dist$, if the probability that for an input $(X,Y)$, Bob outputs some $Z$ where $(X,Y,Z) \notin P$ is at most
$\delta$ (the probability is taken over the randomness of both the distribution and the protocol).

\begin{definition}
  The \emph{communication cost} of a protocol $\prot$ for a problem $P$ on an input
  distribution $\dist$, denoted by $\norm{\prot}$, is the worst-case size of the message
  sent from Alice to Bob in the protocol $\prot$, when the inputs are chosen from the distribution
  $\dist$. \newline The \emph{communication complexity} $\CC{P}{\dist}{\delta}$ of a
  problem $P$ with respect to a distribution $\dist$ is the minimum communication cost of
  a $\delta$-error protocol $\prot$ over $\dist$.
\end{definition}

\paragraph{Information complexity.} There are several possible definitions of information
complexity of a communication problem that have been considered depending on the application (see, e.g.,~\cite{Bar-YossefJKS02,BarakBCR10,ChakrabartiSWY01,Bar-YossefJKS02-S}).  Our definition is tuned
specifically for \emph{one-way protocols}, similar in the spirit of~\cite{Bar-YossefJKS02}
(see also~\cite{JayramW11}).
\begin{definition}
  Consider an input distribution $\dist$ and a protocol $\prot$ (for some problem
  $P$). Let $\bX$ be the random variable for the input of Alice drawn from $\dist$, and let
  $\bprot := \bprot(\bX)$ be the random variable denoting the message sent from Alice to
  Bob \emph{concatenated} with the public randomness $\bR$ used by $\prot$.  The
  information cost $\ICost{\prot}{\dist}$ of a one-way protocol $\prot$ with respect to
  $\dist$ is $I_\dist(\bprot;\bX)$. \newline The \emph{information
    complexity} $\IC{P}{\dist}{\delta}$ of $P$ with respect to a distribution $\dist$ is
  the minimum $\ICost{\prot}{\dist}$ taken over all one-way $\delta$-error protocols
  $\prot$ for $P$ over $\dist$.
\end{definition}

Note that any public coin protocol is a distribution over private coins protocols, run by
first using public randomness to sample a random string $\bR = R$ and then running the
corresponding private coin protocol $\prot^R$. We also use $\bprot^R$ to denote the random
variable of the message sent from Alice to Bob, assuming that the public randomness is
$\bR = R$. We have the following well-known claim. 
\begin{claim}\label{clm:public-random}
	 For any distribution $\dist$ and any protocol $\prot$, let $\bR$ denote the public randomness used in $\prot$; then, $\ICost{\prot}{\dist} = \Ex_{R \sim \bR}\Bracket{I_{\dist}(\bprot^R ; \bX \mid \bR = R)}$. 
\end{claim}
\begin{proof} 
	Let $\bprot = (\bM,\bR)$, where $\bM$ denotes the message sent by Alice and $\bR$ is the public randomness. We have, 
	\begin{align*}
		\ICost{\prot}{\dist} &= I(\bprot ; \bX) = I(\bM, \bR ; \bX) = I(\bR ; \bX) + I(\bM ; \bX \mid \bR) \tag{the chain rule for mutual information, Claim~\ref{clm:it-facts}-(\ref{part:chain-rule})} \\
		&= \Ex_{R \sim \bR}\Bracket{I_{\dist}(\bprot^R ; \bX \mid \bR = R)} \tag{$\bM = \bprot^R$ whenever $\bR = R$ and $I(\bR ; \bX) = 0$ by Claim~\ref{clm:it-facts}-(\ref{part:info-zero})}
	\end{align*}
\end{proof}

The following well-known proposition (see, e.g.,~\cite{ChakrabartiSWY01}) relates communication complexity and information complexity. 
\begin{proposition}\label{prop:cc-ic}
  For every $0 < \delta < 1$ and every distribution $\dist$:
  $\CC{P}{\dist}{\delta} \geq \IC{P}{\dist}{\delta}$.
\end{proposition}
\begin{proof}
  Let $\prot$ be a protocol with the minimum communication complexity for $P$ on $\dist$ and $\bR$ denotes the public randomness of $\prot$;
  using Claim~\ref{clm:public-random}, we can write,
	\begin{align*}
		\IC{P}{\dist}{\delta} &= \Ex_{R \sim \bR}\Bracket{I_{\dist}(\bprot^R ; \bX \mid \bR = R)} \leq \Ex_{R \sim \bR}\Bracket{H_{\dist}(\bprot^R \mid \bR = R)} \\
		&\leq \Ex_{R \sim \bR}\Bracket{\card{\bprot^R}} \leq \norm{\prot} = \CC{P}{\dist}{\delta}
 	\end{align*}
\end{proof}

\section{An ${\Omega}(mn/\alpha)$-Space Lower bound for $\alpha$-Approximate Set Cover}\label{sec:lower-find}

In this section, we prove that the simple $\alpha$-approximation algorithm described in Section~\ref{sec:techniques} is in fact optimal in terms of the space requirement. Formally,

\begin{theorem}\label{thm:find-lower}
  For any $\alpha = o(\frac{\sqrt{n}}{\log{n}})$ and $m = \poly(n)$, any \emph{randomized} single-pass
  streaming algorithm that $\alpha$-approximates the
  set cover problem with probability at least $2/3$ requires ${\Omega}(mn/\alpha)$ bits of
  space.
\end{theorem}

Fix a (sufficiently large) value for $n$, $m = \poly(n)$ (also $m = \Omega(\alpha\log{n})$), and $\alpha = o({\frac{\sqrt{n}}{\log{n}}})$;
throughout this section, \setCoverApx refers to the problem of $\alpha$-approximating the set cover problem for instances with $m+1$
sets\footnote{\label{footnote:m+1}To simplify the exposition, we use $m+1$ instead of $m$ as the number of sets.} defined over the universe
$[n]$ in the one-way communication model, whereby the sets are partitioned between Alice and Bob.

\paragraph{Overview.} We design a hard input distribution $\dista$ for \setCoverApx, whereby Alice is provided with a collection of $m$ sets $S_1,\ldots,S_m$, each of size (roughly) $n/\alpha$ and Bob is
given a \emph{single} set $T$ of size (roughly) $n-2\alpha$. The input to the players are \emph{correlated} such that there exists a set $S_\istar$ in Alice's collection ($\istar$ is unknown to Alice), such that $S_\istar \union T$ covers all elements in $[n]$ except for a single \emph{special element}. This in particular ensures that the optimal set cover size in this distribution is at most $3$ w.h.p. 

On the other hand, we ``hide'' this special element among the $2\alpha$ elements in $\bar{T}$ in a way that if Bob does not have (essentially) full information about Alice's collection, 
he cannot even identify a set of $\alpha$ elements from $\bar{T}$ that contain this special element (w.p strictly more than half). 
This implies that in order for Bob to be sure that he returns a valid set cover, he should additionally cover a majority of $\bar{T}$ with sets \emph{other than} $S_\istar$. We design the distribution
in a way that the sets in Alice's collection are ``far'' from each other and 
hence Bob is forced to use a \emph{distinct} set for (roughly) each element in $\bar{T}$ that he needs to cover with sets other than $S_\istar$. This implies that Bob needs to output a set cover of size $\alpha$ (i.e., 
an $(\alpha/3)$-approximation) to ensure that every element in $[n]$ is covered.

\subsection{A Hard Input Distribution for \setCoverApx}
Consider the following distribution. 

\textbox{Distribution $\dista$. \textnormal{A hard input distribution for \setCoverApx.}}{ 
\medskip\\
	\textbf{Notation.} Let $\FC$ be the collection of all subsets of $[n]$ with cardinality $\frac{n}{10\alpha}$, 
and $\ell := 2\alpha \log{m}$. 
\begin{itemize}
\item \textbf{Alice.} The input of Alice is a collection of $m$ sets $\SA =
  \paren{S_1,\ldots,S_m}$, where for any $i \in [m]$, $S_i$ is a set chosen
  independently and uniformly at random from $\FC$.
\item \textbf{Bob.} Pick an $\istar \in [m]$ (called the \emph{special index}) uniformly at random; the input
  to Bob is a set $T = [n] \setminus E$, where $E$ is chosen uniformly at random from all 
  subsets of $[n]$ with $\card{E} = \ell$ and $\card{E \setminus S_\istar} = 1$.\footnote{Since
    $\alpha = o(\sqrt{n}/\log{n})$ and $m = \poly(n)$, the size of $E$ is strictly smaller than the size of $S_\istar$.} 
\end{itemize}
}

The claims below summarize some useful properties of the distribution $\dista$. 
\begin{claim}\label{clm:opt-size}
  For any instance $(\SA,T) \sim \dista$, with probability $1- o(1)$,
  $\opt(\SA,T) \leq 3$.
\end{claim}
\begin{proof}
  Let $\estar$ denote the element in $E \setminus S_\istar$.  $\SAistar$ contains $m-1$
  random subsets of $[n]$ of size $n/10\alpha$, drawn independent of the choice of
  $\estar$. Therefore, each set in $\SAistar$ covers $\estar$ with probability
  $1/10\alpha$. The probability that none of these $m-1$ sets covers $\estar$ is at most
  \[
    (1-1/10\alpha)^{m-1} \leq  (1-1/10\alpha)^{\Omega(\alpha \log{n})} \leq \exp(-\Omega(\alpha \log{n})/10\alpha)
    = o(1)
  \]
  Hence, with probability $1-o(1)$, there is at least one set $S \in \SA^{-\istar}$ that covers $\estar$. Now, it is
  straightforward to verify that $(S_\istar, T, S)$ form a valid set cover.
\end{proof}

\begin{lemma}\label{lem:cover-size}
  With probability $1-o(1)$, no collection of $3\alpha$ sets from $\SAistar$ covers more
  than $\ell/2$ elements of $E$.
\end{lemma}
\begin{proof}
  Recall that the sets in $\SAistar$ and the set $E$ are chosen independent of each other. For each
  set $S \in \SAistar$ and for each element $e \in E$, we define an indicator binary random
  variable $\bX_e$, where $\bX_e = 1$ iff $e \in S$.  Let $\bX:= \sum_{e} \bX_e$, which is
  the number of elements in $E$ covered by $S$. We have,
  \[
    \Ex[\bX] = \sum_{e} \Ex[\bX_e] = \frac{\card{E}}{10\alpha} = \frac{\log{m}}{5}
  \]
  Moreover, the variables $\bX_e$ are negatively correlated since for any set $S' \subseteq E$, 
  \begin{align*}
  	\Pr\paren{\bigwedge_{e \in S'} \bX_e = 1} = {{{n-\card{S'}}\choose{{n\over10\alpha} - \card{S'}}}\over{{n}\choose{n\over10\alpha}}} &= 
	\frac{\paren{\frac{n}{10\alpha}} \cdot \paren{\frac{n}{10\alpha} - 1} \ldots \paren{\frac{n}{10\alpha} - \card{S'} + 1}}{\paren{n} \cdot \paren{n-1} \ldots \paren{n- \card{S'} + 1}} \\ 
	&\leq \paren{\frac{1}{10\alpha}}^{\card{S'}} = \prod_{e \in S'} \Pr\paren{\bX_e = 1} 
  \end{align*}
  Hence, by the extended Chernoff bound (see Section~\ref{sec:prelim}),
  \[
    \PR{\bX \geq \frac{\log{m}}{3}} = o(\frac{1}{m})
  \]
  Therefore, using union bound over all $m-1$ sets in $\SAistar$, with probability $1-o(1)$,
  no set in $\SAistar$ covers more than $\log{m}/3$ elements in $E$, which implies that
  any collection of $3\alpha$ sets can only cover up to
  $3\alpha \cdot \log{m}/3 = \ell/2$ elements in $E$.
\end{proof}

\subsection{The Lower Bound for the Distribution $\dista$}

In order to prove our lower bound for \setCoverApx on $\dista$, we define an intermediate
communication problem which we call the \emph{Trap} problem. 
\begin{problem}[\emph{Trap} problem]
  Alice is given a set $S \subseteq [n]$ and Bob is given a set $E \subseteq [n]$ such
  that $E \setminus S = \set{e^*}$; Bob needs to output a set $L \subseteq E$
  with $\card{L} \le \card{E}/2$ such that $e^* \in L$.
\end{problem} 

In the following, we use \Trap to refer to the trap problem with $\card{S} = n/10\alpha$ and $\card{E} = \ell = 2\alpha\log{m}$ (notice the similarity to the parameters in $\dista$).
We define the following distribution $\distT$ for \Trap. Alice is given a set
$S \in_R \FC$ (recall that $\FC$ is the collection of all subsets of $[n]$ of size
$n/10\alpha$) and Bob is given a set $E$ chosen uniformly at random from all sets that
satisfy $\card{E \setminus S} = 1$ and $\card{E} = 2\alpha \log{m}$.  We first use a
direct sum style argument to prove that under the distributions $\dista$ and $\distT$,
information complexity of solving \setCoverApx is essentially equivalent to solving $m$ copies of \Trap. Formally,

\begin{lemma}\label{lem:sc-trap}
  For any constant $\delta < 1/2$,
  $\IC{\setCoverApx}{\dista}{\delta} \geq m \cdot \IC{\Trap}{\distT}{\delta+o(1)}$.
\end{lemma}
\begin{proof}
  Let $\protsc$ be a $\delta$-error protocol for \setCoverApx; we design a $\delta'$-error
  protocol $\prottrap$ for solving \Trap over $\distT$ with parameter $\delta' = \delta + o(1)$
  such that the information cost of $\prottrap$ on $\distT$ is at most $\frac{1}{m} \cdot
  \ICost{\protsc}{\dista}$.
  The protocol $\prottrap$ is as follows.  
  
  \textbox{Protocol
    $\prottrap$. \textnormal{The protocol for solving \Trap using a protocol $\protsc$ for \setCoverApx.}}{
    \medskip \\
    \textbf{Input:} An instance $(S,E) \sim \distT$. 
    \textbf{Output:} A set $L$ with $\card{L} \leq \card{E}/2$, such that $e^* \in L$.  \\
    \algline
\begin{enumerate}
\item Using \emph{public randomness}, the players sample an index $\istar \in [m]$ uniformly at random.
\item Alice creates a tuple $\SS = (S_1,\ldots,S_m)$ by assigning $S_\istar = S$ and sampling
  each remaining set uniformly at random from $\FC$ using \emph{private randomness}. Bob creates a
  set $T:= \bar{E}$.
\item The players run the protocol $\protsc$ over the input $(\SS,T)$.
\item Bob computes the set $L$ of all elements in $E = \bar{T}$ whose certificate (i.e., the set used to cover them) is not $S_\istar$, and outputs $L$.
\end{enumerate}
} 

We first argue the correctness of $\prottrap$ and then bound its information cost. To
argue the correctness, notice that the distribution of instances of \setCoverApx
constructed in the reduction is exactly $\dista$.  Consequently, it follows from
Claim~\ref{clm:opt-size} that, with probability $1-o(1)$, any $\alpha$-approximate set cover can have at most
$3\alpha$ sets. Let $\Salpha$ be the set cover computed by Bob minus the sets $S_{\istar}$ and $T$. 
As $\estar \in E = \bar{T}$ and moreover is \emph{not} in $S_{\istar}$, it follows that $\estar$ should be covered by some set in $\Salpha$.  
This means that the set $L$ that is output by Bob contains $\estar$. 
Moreover, by Lemma~\ref{lem:cover-size}, the number of
elements in $E$ covered by the sets in $\Salpha$ is at most
$\ell/2$ w.p. $1-o(1)$. Hence, $\card{L} \leq \ell/2 = \card{E}/2$. This implies that:
\begin{align*}
	\Pr_{\distT}\Paren{\prottrap \errs} \leq \Pr_{\dista}\Paren{\protsc \errs} + o(1) \leq \delta + o(1)
\end{align*}  

We now bound the information cost of $\prottrap$. Let $\bI$ be the random variable for the choice of $i^* \in [m]$ in the protocol $\prottrap$ (which is uniform in
$[m]$). Using Claim~\ref{clm:public-random}, we have,
\begin{align*}
	\ICost{\prottrap}{\distT} &= \Ex_{i \sim \bI} \Bracket{I_{\distT}\Paren{\bprottrap^i;\bS \mid \bI = i}} 
	= \frac{1}{m} \cdot \sum_{i=1}^{m} I_{\distT} \Paren{\bprotsc ; \bS_i \mid \bI = i} \\
	&= \frac{1}{m} \cdot \sum_{i=1}^{m} I_{\dista} \Paren{\bprotsc ; \bS_i \mid \bI = i} 
	= \frac{1}{m} \cdot \sum_{i=1}^{m} I_{\dista} \Paren{\bprotsc ; \bS_i}
\end{align*}
where the last two equalities hold since $(i)$ the joint distribution of $\bprotsc$ and $\bS_i$
conditioned on $\bI = i$ under $\distT$ is equivalent to the one under $\dista$, and $(ii)$ the random variables $\bprotsc$ and $\bS_i$ are
independent of the event $\bI = i$ (by the definition of $\dista$) and hence we can ``drop'' the conditioning on this event (by Claim~\ref{clm:it-facts}-(\ref{part:info-event})).

We can further derive, 
\begin{align*}
	\ICost{\prottrap}{\distT} &= \frac{1}{m} \cdot \sum_{i=1}^{m} I_{\dista} \Paren{\bprotsc ; \bS_i} 
	\leq \frac{1}{m} \cdot \sum_{i=1}^{m} I_{\dista} \Paren{\bprotsc ; \bS_i \mid \bSS^{<i}} 
\end{align*}
	The inequality holds since $\bS_i$ and $\bSS^{<i}$ are independent and conditioning on independent variables can only increase the mutual information (i.e., Claim~\ref{clm:info-increase}). 
	Finally, 
\begin{align*}
	\ICost{\prottrap}{\distT} &\leq \frac{1}{m} \cdot \sum_{i=1}^{m} I_{\dista} \Paren{\bprotsc ; \bS_i \mid \bSS^{<i}} 
	= \frac{1}{m} \cdot  I_{\dista} \Paren{\bprotsc; \bSS} = \frac{1}{m} \cdot \ICost{\protsc}{\dista}
\end{align*}
where the first equality is by the chain rule for mutual information (see Claim~\ref{clm:it-facts}-(\ref{part:chain-rule})).
\end{proof}

Having established Lemma~\ref{lem:sc-trap}, our task now is to lower bound the information complexity of
\Trap over the distribution $\distT$.  We prove this lower bound using
a novel reduction from the well-known \emph{Index} problem, denoted by $\Indexnk$. In $\Indexnk$ over the distribution $\distI$, Alice is given a set $A \subseteq [n]$ of size $k$ chosen
uniformly at random and Bob is given an element $a$ such that w.p. $1/2$
$a \in_R A$ and w.p. $1/2$ $a \in_R [n] \setminus A$; Bob needs to
determine whether $a \in A$ (the \Yes case) or not (the \No case).

We remark that similar distributions for \Indexnk have been previously studied in the
literature (see, e.g., \cite{SaglamThesis2011}, Section~3.3). For the sake of completeness, we
provide a self-contained proof of the following lemma in Appendix~\ref{app:index}.

\begin{lemma}\label{lem:index}
	For any $k < n/2$, and any constant $\delta' < 1/2$, $\IC{\Indexnk}{\distI}{\delta'} = \Omega(k)$. 
\end{lemma}
	
Using Lemma~\ref{lem:index}, we prove the following lemma, which is the key part of the proof.

\begin{lemma}\label{lem:trap}
	For any constant $\delta < 1/2$, $\IC{\Trap}{\distT}{\delta} = \Omega(n/\alpha)$. 
\end{lemma}

\begin{proof}
  Let $k = n/10\alpha$; we design a $\delta'$-error protocol $\protindex$ for \Indexnk using any $\delta$-error protocol
  $\prottrap$ (over $\distT$) as a subroutine, for some constant $\delta' < 1/2$.
 
  \textbox{Protocol $\protindex$. \textnormal{The protocol for reducing \Indexnk to \Trap.}}{\medskip \\
    \textbf{Input:} An instance $(A,a) \sim \distI$. \textbf{Output:} \Yes if $a \in A$ and \No otherwise.  \\ \algline
    \begin{enumerate}
    \item Alice picks a set $B \subseteq A$ with $\card{B} = \ell - 1$ uniformly at random using \emph{private randomness}.
    \item To invoke the protocol $\prottrap$, Alice creates a set $S := A$ and sends the message $\prottrap(S)$, along with the set $B$ to Bob.
    \item \label{line:lucky} If $a \in B$, Bob outputs \Yes and terminates the protocol. 
    \item Otherwise, Bob constructs a set $E = B \cup \set{a}$ and computes $L:= \prottrap(S,E)$ using the message received from Alice.
    \item If $a \in L$, Bob outputs \No, and otherwise outputs \Yes.
    \end{enumerate}
  }
      
  We should note right away that the distribution of instances for \Trap defined in the previous reduction does \emph{not} match $\distT$. 
  Therefore, we need a more careful argument to establish the correctness of the reduction.
	
  We prove this lemma in two claims; the first claim establishes the correctness
  of the reduction and the second one proves an upper bound on the information cost of
  $\protindex$ based on the information cost of $\prottrap$.

\begin{claim}\label{clm:reduction-correctness}
  $\protindex$ is a $\delta'$-error protocol for $\Indexnk$ over $\distI$ for the parameter
  $k = n/10\alpha$ and a constant $\delta' < 1/2$.
\end{claim}
\begin{proof}
  Let $\bR$ denote the private coins used by Alice to construct the set $B$.  Also,
  define $\distIY$ (resp. $\distIN$) as the distribution of \Yes instances (resp. \No
  instances) of $\distI$.  We have,
  \begin{align}
    \Pr_{\distI,\bR} \Paren{\protindex \errs} = \frac{1}{2} \cdot \Pr_{\distIY,\bR}\Paren{\protindex \errs} + \frac{1}{2} \cdot \Pr_{\distIN,\bR}\Paren{\protindex \errs} \label{eq:err-prob}
  \end{align}
  Note that we do not consider the randomness of the protocol $\prottrap$ (used in construction of $\protindex$) as it is independent of the
  randomness of the distribution $\distI$ and the private coins $\bR$.  We now bound each term in Equation~(\ref{eq:err-prob})
  separately. We first start with the easier case which is the second term.
	
  The distribution of instances $(S,E)$ for \trap created in the reduction by
  the choice of $(A,a) \sim \distIN$ and the randomness of $\bR$, is the same as the
  distribution $\distT$. Moreover, in this case, the output of $\protindex$ would be
  wrong iff $a \in E \setminus S$ (corresponding to the element $e^*$ in \Trap) does not belong to the set $L$ output by $\prottrap$. Hence,
  \begin{align}
    \Pr_{\distIN,\bR}\Paren{\protindex \errs} = \Pr_{\distT}\Paren{\prottrap \errs} \leq \delta \label{eq:err-trap}
  \end{align} 
  We now bound the first term in Equation~(\ref{eq:err-prob}). Note that when
  $(A,a) \sim \distIY$, there is a small chance that $\protindex$ is ``lucky" and $a$
  belongs to the set $B$ (see Line~(\ref{line:lucky}) of the protocol). Let this event be
  $\event$. Conditioned on $\event$, Bob outputs the correct answer with probability $1$;
  however note that probability of $\event$ happening is only $o(1)$.  Now suppose
  $\event$ does not happen. In this case, the distribution of instances $(S,E)$ created by
  the choice of $(A,a) \sim \distIY$ (and randomness of $\bR$) does \emph{not} match
  the distribution $\distT$. However, we have the following important property: Given
  that $(S,E)$ is the instance of \Trap created by choosing $(A,a)$ from $\distIY$ and sampling $\ell - 1$ random elements of $A$ (using $\bR$),
 the element $a$ is \emph{uniform} over the set $E$. In other words, knowing $(S,E)$ does not reveal any information about the element $a$. 
	
  Note that since $(S,E)$ is not chosen according to the distribution $\distT$ (actually
  it is not even a ``legal'' input for \trap), it is possible that $\prottrap$
  terminates, outputs a non-valid set, or outputs a set $L \subseteq E$.  Unless
  $L \subseteq E$ (and satisfies the cardinality constraint), Bob is always able to
  determine that $\prottrap$ is not functioning correctly and hence outputs \Yes (and errs
  with probability at most $\delta < 1/2$). However, if $L \subseteq E$, Bob would not know whether the input to $\prottrap$ is legal or not. In the following, we explicitly
  analyze this case.
	
  In this case, $L$ is a subset of $E$ chosen by the (inner) randomness of $\prottrap$ for a fixed $S$ and $E$ and moreover
  $\card{L} \le \card{E}/2$ (by definition of \Trap). The probability that $\protindex$ errs in this case is exactly equal to the
  probability that $a \in L$. However, as stated before, for a fixed $(S,E)$, the choice of $L$ is independent of the choice of $a$ and
  moreover, $a$ is uniform over $E$; hence $a \in L$ happens with probability at most $1/2$. Formally, (here, $\bR^{\trap}$ denotes the inner randomness of $\prottrap$)
  \begin{align*}
  	\Pr_{\distIY,\bR}(\protindex \errs \mid \bar{\event}) &= \Pr_{\distIY,\bR}\Paren{a \in L=\prottrap(\bS,\bE) \mid \bar{\event}} \\
	&= \Ex_{(S,E) \sim (\bS,\bE) \mid \bar{\event}} \Ex_{R^{\trap} \sim \bR^{\trap}} \Bracket{\Pr_{\distIY,\bR}\Paren{a \in L \mid \bS = S, \bE = E, \bR^{\trap} = R^{\trap}, \bar{\event}}}
	\tag{$L = \prottrap(S,E)$ is a fixed set conditioned on $(S,E,R^{\trap})$} \\
	&=  \Ex_{(S,E) \sim (\bS,\bE) \mid \bar{\event}} \Ex_{R^{\trap} \sim \bR^{\trap}} \Bracket{\frac{\card{L}}{\card{E}}} \tag{$a$ is uniform on $E$ conditioned on $(S,E,R^{\trap})$ and $\bar{\event}$}
  \end{align*}
  Hence, we have, $\Pr_{\distIY,\bR}(\protindex \errs \mid \bar{\event}) \leq \frac{1}{2}$, since by definition, for any output set $L$, $\card{L} \leq \card{E}/2$. 
  
  As stated earlier, whenever $\event$ happens, $\protindex$ makes no error; hence, 
	\begin{align}
		\Pr_{\distIY,\bR}(\protindex \errs) = \Pr_{\distIY,\bR}(\bar{\event}) \cdot \Pr_{\distIY,\bR}(\protindex \errs \mid \bar{\event}) \leq \frac{1-o(1)}{2}
		 \label{eq:err-non-trap}
	\end{align}
	Finally, by plugging the bounds in Equations~(\ref{eq:err-trap},\ref{eq:err-non-trap}) in Equation~(\ref{eq:err-prob}) and assuming
        $\delta$ is bounded away from $1/2$, we have,
	\begin{align*}
		\Pr_{\distI,\bR} (\protindex \errs) \leq \frac{1}{2}\cdot\frac{1-o(1)}{2} + \frac{1}{2} \cdot \delta = \frac{1-o(1)}{4} + \frac{\delta}{2} \leq \frac{1}{2} - \eps
	\end{align*}
	for some constant $\eps$ bounded away from $0$.
\end{proof}

We now bound the information cost of $\protindex$ under $\distI$. 
\begin{claim}\label{clm:reduction-info}
	$\ICost{\protindex}{\distI} \leq \ICost{\prottrap}{\distT} + O(\ell\log{n})$.
\end{claim}
\begin{proof}
	We have,
	\begin{align*}
		\ICost{\protindex}{\distI} &= I_{\distI}\Paren{\bprotindex(\bA) ; \bA} \\
		&= I_{\distI}\Paren{\bprottrap(\bS),\bB;\bA} \\
		&=I_{\distI}\Paren{\bprottrap(\bS); \bA} + I_{\distI}\Paren{\bB ; \bA \mid \bprottrap(\bS)}  
		 \tag{the chain rule for mutual information, Claim~\ref{clm:it-facts}-(\ref{part:chain-rule})} \\
		 &\leq I_{\distI}\Paren{\bprottrap(\bS);\bA }  + H_{\distI}\Paren{\bB \mid  \bprottrap(\bS)} \\
		 &\leq I_{\distI}\Paren{\bprottrap(\bS); \bA} + O(\ell\log{n}) 
		 \tag{$\card{\bB} = O(\ell\log{n})$ and Claim~\ref{clm:it-facts}-(\ref{part:uniform})} \\
		 &= I_{\distI}\Paren{\bprottrap(\bS); \bS} + O(\ell\log{n}) \tag{$\bA = \bS$ as defined in $\protindex$}\\
		 &= I_{\distT}\Paren{\bprottrap(\bS);\bS} + O(\ell\log{n}) \tag{the joint distribution of $(\bprottrap(\bS),\bS)$ is identical under $\distI$ and $\distT$} \\
		 &= \ICost{\prottrap}{\distT} + O(\ell\log{n})
	\end{align*}
\end{proof}

	The lower bound now follows from Claims~\ref{clm:reduction-correctness} and \ref{clm:reduction-info}, and Lemma~\ref{lem:index} for the parameters $k = \card{S} = \frac{n}{10\alpha}$ and $\delta' <1/2$,
	and using the fact that $\alpha = o(\sqrt{n}/\log{n})$, $\ell = 2\alpha\log{m}$, and $m = \poly(n)$, and hence $\Omega(n/\alpha) = \omega(\ell\log{n})$. 

\end{proof}

To conclude, by Lemma~\ref{lem:sc-trap} and Lemma~\ref{lem:trap}, for any set of parameters $\delta < 1/2$, $\alpha = o(\frac{\sqrt{n}}{\log{n}})$, and $m = \poly(n)$, 
\begin{align*}
\IC{\setCoverApx}{\dista}{\delta} \geq m\cdot\Paren{\Omega(n/\alpha)} =  {\Omega}(mn/\alpha) 
\end{align*}
Since the information complexity is a lower bound on the communication complexity (Proposition~\ref{prop:cc-ic}), we have,

\begin{theorem}\label{thm:two-find-adv}
	For any constant $\delta < 1/2$, $\alpha = o(\frac{\sqrt{n}}{\log{n}})$, and $m = \poly(n)$, 
	\begin{align*}
		\CC{\setCoverApx}{\dista}{\delta} = \Omega(mn/\alpha)
	\end{align*}
\end{theorem}

Finally, since one-way communication complexity is also a lower bound on 
the space complexity of single-pass streaming algorithms, we obtain Theorem~\ref{thm:find-lower} as a corollary of Theorem~\ref{thm:two-find-adv}.

\section{An $\Ot(mn/\alpha^2)$-Space Upper Bound for $\alpha$-Estimation}\label{sec:general}

In this section, we show that if we are only interested
in estimating the \emph{size} of {a minimum set cover} (instead of finding the actual
sets), we can bypass the $\Omega(mn/\alpha)$ lower bound established in
Section~\ref{sec:lower-find}. In fact, we prove this upper bound for the more general problem of estimating
the optimal solution of a \emph{covering integer program} (henceforth,
\emph{covering ILP}) in the streaming setting. 
 
 A covering ILP can be formally defined as follows. 
\begin{align*}
  \text{min~}  c \cdot  x
  \text{~~s.t.~~~} Ax \geq b
\end{align*}
where $A$ is a matrix with dimension $n \times m$, $b$ is a vector of dimension $n$, $c$
is a vector of dimension $m$, and $x$ is an $m$-dimensional vector of non-negative integer
variables. Moreover, all coefficients in $A,b,$ and $c$ are also
{non-negative integers}.  We denote this linear program by $\CILP(A,b,c)$. We use
$\amax$, $\bmax$, and $\cmax$, to denote the \emph{largest} entry of, respectively, the
matrix $A$, the vector $b$, and the vector $c$. Finally, we define the \emph{optimal
  value} of the $\Ins:=\CILP(A,b,c)$ as $c \cdot x^*$ where $x^*$ is the \emph{optimal}
solution to $\Ins$, and denote it by $\opt := \opt(\Ins)$.

We consider the following streaming setting for covering ILPs. The input to a streaming
algorithm for an instance $\Ins:= \CILP(A,b,c)$ is the $n$-dimensional vector $b$, and a
stream of the $m$ columns of $A$ presented one by one, where the $i$-th column of $A$,
$A_i$, is presented along with the $i$-th entry of $c$, denoted by $c_i$ (we will refer to $c_i$ as
the weight of the $i$-th column). It is easy to see that this streaming setting for
covering ILPs captures, as special cases, the \emph{set cover} problem, the \emph{weighted set cover}
problem, and the \emph{set multi-cover} problem.  We prove the following theorem for
$\alpha$-estimating the optimal value of a covering ILPs in the streaming setting.

\begin{theorem}\label{thm:general}
  There is a randomized algorithm that given a parameter $\alpha \geq 1$, for any
  instance $\Ins:=\CILP(A,b,c)$ with $\poly(n)$-bounded entries, makes a single pass over a stream
  of columns of $A$ (presented in an arbitrary order), and outputs an $\alpha$-estimation
  to $\opt(\Ins)$ w.h.p. using space $\Ot\paren{(mn/\alpha^2)\cdot\bmax + m + n\bmax}$ bits.
  
  In particular, for the \emph{weighted set cover} problem with $\poly(n)$ bounded weights
  and $\alpha \leq \sqrt{n}$, the space complexity of this algorithm is
  $\Ot(mn/\alpha^2 + n)$.\footnote{Note that $\Omega(n)$ space is necessary to even determine whether or not a given instance is feasible.}
\end{theorem}

To prove Theorem~\ref{thm:general}, we design a general approach based on sampling constraints of a covering ILP instance. The goal is to show that 
if we sample (roughly) $1/\alpha$ fraction of the constraints from an instance $\Ins:=\CILP(A,b,c)$, then the optimum value of the resulting covering ILP, denoted by $\RIns$, 
is a good estimator of $\opt(\Ins)$. Note that in general, this may not be the case; simply consider a weighted set cover instance that contains an element $e$ which is only covered by a singleton set 
of weight $W$ (for $W \gg m$) and all the remaining sets are of weight $1$ only. Clearly, $\opt(\RIns) \ll \opt(\Ins)$ as long as $e$ is not sampled in $\RIns$, which happens w.p. $1-1/\alpha$. 

To circumvent this issue, we define a notion of \emph{cost} for covering ILPs which,
informally, is the minimum value of the objective function if the goal is to only satisfy
a single constraint (in the above example, the cost of that weighted set cover instance is
$W$). This allows us to bound the loss incurred in the process of estimation by sampling based on the
cost of the covering ILP.

Constraint sampling alone can only reduce the space requirement by a factor
of $\alpha$, which is not enough to meet the bounds given in
Theorem~\ref{thm:general}. Hence, we combine it with a \emph{pruning} step to sparsify the columns
in $A$ before performing the sampling. We should point out that as columns are weighted,
the pruning step needs to be sensitive to the weights.

In the rest of this section, we first introduce our \emph{constraint sampling lemma}
(Lemma~\ref{lem:sub-sample-save-opt}) and prove its correctness, and then provide our
algorithm for Theorem~\ref{thm:general}.

\subsection{Covering ILPs and Constraint Sampling Lemma}\label{sec:csl}

In this section, we provide a general result for estimating the optimal value of a
Covering ILP using a sampling based approach.  For a vector $v$, we will use $v_i$ to
denote the $i$-th dimension of $v$. For a matrix $A$, we will use $A_i$ to denote the
$i$-th column of $A$, and use $a_{j,i}$ to denote the entry of $A$ at the $i$-th column
and the $j$-th row (to match the notation with the set cover problem, we use $a_{j,i}$
instead of the standard notation $a_{i,j}$).

For each constraint $j \in [n]$ (i.e., the $j$-th constraint) of a covering ILP instance
$\Ins:= \CILP(A,b,c)$, we define the \emph{cost} of the constraint $j$, denoted by
$\cost(j)$, as,
\begin{align*}
	\cost(j) := &\min_{x} c \cdot x ~~\text{ s.t } \sum_{i=1}^{m} a_{j,i} \cdot x_i \geq b_j
\end{align*} 
which is the \emph{minimum solution value} of the objective function for satisfying the constraint
$j$. Furthermore, the \emph{cost} of $\Ins$, denoted by $\cost(\Ins)$, is defined to be
\[
	\cost(\Ins) :=  \max_{j \in [n]} \cost(j)
\]
Clearly, $\cost(\Ins)$ is a lower bound on $\opt(\Ins)$.
 
\paragraph{Constraint Sampling.} Given any instance of covering ILP
$\Ins := \CILP(A,b,c)$, let $\RIns$ be a covering ILP instance $\CILP(A,\widetilde{b},c)$
obtained by setting $\widetilde{b}_j := b_j$ with probability $p$, and
$\widetilde{b}_j := 0$ with probability $1-p$, for each dimension $j \in [n]$ of $b$
independently.  Note that setting $\widetilde{b}_j := 0$ in $\RIns$ is equivalent to removing the $j$-th
constraint from $\Ins$, since all entries in $\Ins$ are non-negative. Therefore,
intuitively, $\RIns$ is a covering ILP obtained by sampling (and keeping) the constraints
of $\Ins$ with a sampling rate of $p$.

We establish the following lemma which asserts that $\opt(\RIns)$ is a good estimator
of $\opt(\Ins)$ (under certain conditions). As $\opt(\RIns) \leq \opt(\Ins)$ trivially holds (removing
constraints can only decrease the optimal value), it suffices to give a lower bound on
$\opt(\RIns)$.
  
\begin{lemma}[Constraint Sampling Lemma]\label{lem:sub-sample-save-opt}
  Fix an $\alpha \geq 32\ln{n}$; for any covering ILP $\Ins$ with $n$ constraints, suppose $\RIns$ is obtained from $\Ins$ by sampling each
  constraint with probability $p:= {4\ln{n} \over \alpha}$; then
  \[ \Pr\Paren{\opt(\RIns) + \cost(\Ins) \geq \frac{\opt(\Ins)}{8\alpha} } \geq \frac{3}{4}\]
\end{lemma}
\begin{proof} 
  Suppose by contradiction that the lemma statement is false and throughout the proof let $\Ins$ be any instance where w.p. at least $1/4$, $\opt(\RIns) + \cost(\Ins)  < {\opt(\Ins) \over 8\alpha} $ 
  (we denote this event by $\event_1(\RIns)$, or shortly $\event_1$). We
  will show that in this case, $\Ins$ has a feasible solution with a value smaller than $\opt(\Ins)$.  To continue,
  define $\event_2(\RIns)$ (or $\event_2$ in short) as the event that $\opt(\RIns) < {\opt(\Ins) \over 8\alpha}$. Note
  that whenever $\event_1$ happens, then $\event_2$ also happens,  hence $\event_2$ happens
  w.p. at least $1/4$.

  For the sake of analysis, suppose we repeat, for $32\alpha$ times, the procedure of sampling each constraint of $\Ins$ independently with
  probability $p$, and obtain $32\alpha$ covering ILP instances $S:=\set{\RIns^1, \ldots, \RIns^{32\alpha}}$. Since $\event_2$ happens with
  probability at least $1/4$ on each instance $\RIns$, the expected number of times that $\event_2$ happens for instances in $S$ is at least
  $8\alpha > 12\ln{n}$. Hence, by the Chernoff bound, with probability at least $1 - 1/n$, $\event_2$ happens on at least $4\alpha$ of
  instances in $S$. Let $T \subseteq S$ be a set of $4\alpha$ sampled instances for which $\event_2$ happens.  In the following, we show that if
  $\Ins$ has the property that $\Pr(\event_1(\Ins_R)) \geq 1/4$, then w.p. at least $1 - 1/n$, every constraint in $\Ins$ appears in at
  least one of the instances in $T$. Since each of these $4\alpha$ instances admits a solution of value at most
  ${\opt(\Ins) \over 8 \alpha}$ (by the definition of $\event_2$), the ``max'' of their solutions, i.e., the vector obtained by setting the
  $i$-th entry to be the largest value of ${x}_i$ among all these solutions, gives a feasible solution to $\Ins$ with value at most
  ${4\alpha\cdot {\opt(\Ins) \over 8\alpha} = {\opt(\Ins) \over 2}}$; a contradiction.
  
  We use ``$j \in \RIns$'' to denote the event that the constraint $j$ of $\Ins$ is
  sampled in $\RIns$, and we need to show that w.h.p. for all $j$, there exists an instance
  $\RIns \in T$ where $j \in \RIns$. We establish the following claim.
    
    \begin{claim}\label{clm:high-prob}
   	For any $j \in [n]$, $\Pr\Paren{j \in \RIns \mid \event_2(\RIns)} \geq {\ln{n} \over 2\alpha}$.
    \end{claim}
    Before proving Claim~\ref{clm:high-prob}, we show how this claim would imply the
    lemma.  By Claim~\ref{clm:high-prob}, for each of the $4\alpha$ instances
    $\RIns \in T$, and for any $j \in [n]$, the probability that the constraint $j$ is sampled in $\RIns$ is
    at least ${\ln{n} \over 2\alpha}$.  Then, the probability that $j$ is sampled in none
    of the $4\alpha$ instances of $T$ is at most:
    \begin{align*}
      \Paren{1-{\ln{n} \over 2\alpha}}^{4\alpha} \leq \exp(-2\ln{n}) = \frac{1}{n^2}
    \end{align*}
    Hence, by union bound, w.p. at least $1 - 1/n$, every constraint appears in at least
    one of the instances in $T$, and this will complete the proof. It remains to
    prove Claim~\ref{clm:high-prob}.

   \begin{proof}[Proof of Claim~\ref{clm:high-prob}]
    Fix any  $j \in [n]$; by Bayes rule, 
    \begin{align*}
      \Pr\Paren{j \in \RIns \mid \event_2(\RIns) } = {\Pr\Paren{\event_2(\RIns)  \mid j \in \RIns} \cdot \PR{j \in
          \RIns} \over \Pr\Paren{\event_2 (\RIns)}}
    \end{align*}
    Since $\Pr\Paren{\event_2 (\RIns) } \leq 1$ and $\PR{j \in \RIns} = p = {4\ln{n} \over \alpha}$, we have, 
    \begin{align}
      \Pr\Paren{j \in \RIns \mid \event_2(\RIns)} \geq \Pr\Paren{\event_2(\RIns)  \mid j \in \RIns} \cdot \frac{4\ln{n}}{\alpha} \label{eq:event-2}
    \end{align}
    and it suffices to establish a lower bound of $1/8$ for
    $\Pr\Paren{\event_2(\RIns) \mid j \in \RIns}$.

    Consider the following probabilistic process (for a fixed $j \in [n]$): we first remove the constraint $j$ from
    $\Ins$ (w.p.  $1$) and then sample each of the remaining constraints of $\Ins$
    w.p. $p$. Let $\RIns'$ be an instance created by this process.  We prove
    $\PR{\event_2(\RIns) \mid j \in \RIns} \ge 1/8$  in two steps by
    first showing that the probability that $\event_1$ happens to $\RIns'$ (i.e.,
    $\PR{\event_1(\RIns')}$) is at least $1/8$, and then use a coupling argument to prove that
    $\PR{\event_2(\RIns) \mid j \in \RIns} \geq \PR{\event_1(\RIns')}$.

    We first show that $\PR{\event_1(\RIns')}$ (which by definition is the
    probability that $ \opt(\RIns') + \cost(\Ins) \leq {\opt(\Ins) \over 16\alpha} $) is
    at least $1/8$. To see this, note that the probability that $\event_1$ happens to
    $\RIns'$ is equal to the probability that $\event_1$ happens to $\RIns$ conditioned on
    $j$  not being sampled (i.e.,
    $\PR{\event_1(\RIns) \mid j \notin \RIns}$).  Now, if we expand
    $\PR{\event_1(\RIns)}$,
    \begin{align*}
    \Pr\Paren{\event_1(\RIns)} &= \PR{j \in \RIns}\Pr\Paren{\event_1(\RIns) \mid j
      \in \RIns} + \PR{j \notin \RIns} \Pr\Paren{\event_1(\RIns) \mid j \notin \RIns} \\
      &\leq \PR{j \in \RIns} + \Pr\Paren{\event_1(\RIns) \mid j \notin \RIns} = p + \Pr\Paren{\event_1(\RIns')} 
    \end{align*}
    As $\PR{\event_1(\RIns)} \geq 1/4$ and $p = {4\ln{n} \over \alpha} \leq 1/8$ (since $\alpha \ge 32\ln{n}$), we have,
    \[
    1/4 \le 1/8 + \Pr\Paren{\event_1(\RIns')}
    \]
    and therefore, $\PR{\event_1(\RIns')} \ge 1/8$.

    It remains to show that
    $\PR{\event_2(\RIns)\mid j \in \RIns} \ge \PR{\event_1(\RIns')}$.  To see this, note
    that conditioned on $j \in \RIns$, the distribution of sampling all constraints other
    than $j$ is exactly the same as the distribution of $\RIns'$.  Therefore, for any
    instance $\RIns'$ drawn from this distribution, there is a unique instance $\RIns$
    sampled from the original constraint sampling distribution \emph{conditioned} on
    $j \in \RIns$.  For any such ($\RIns'$, $\RIns$) pair, we have
    $\opt(\RIns) \leq \opt(\RIns') + \cost(j)$ $(\leq \opt(\RIns') + \cost(\Ins))$ since
    satisfying the constraint $j \in \RIns$ requires increasing the value of the objective function in
    $\RIns'$ by at most $\cost(j)$. Therefore if
    $\opt(\RIns') + \cost(\Ins) \le {\opt \over 8\alpha}$ (i.e., $\event_1$ happens to
    $\RIns'$), then $\opt(\RIns) \le {\opt \over 8\alpha}$ (i.e., $\event_2$ happens to
    $\RIns$ conditioned on $j \in \RIns$).  Hence,
    \begin{align*}
      \Pr\Paren{\event_2(\RIns) \mid j \in \RIns}
      \geq \Pr\Paren{ \event_1(\RIns')} \geq 1/8
    \end{align*}
    
    Plugging in this bound in Equation~(\ref{eq:event-2}), we obtain that
    $\PR{j \in \RIns \mid \event_2} \geq {\ln{n} \over 2\alpha}$.
  \end{proof}
\end{proof}

\subsection{An $\alpha$-estimation of Covering ILPs in the Streaming Setting}\label{sec:ilp-est}

We now prove Theorem~\ref{thm:general}.  Throughout this section, for simplicity of exposition, we assume that $\alpha \geq 32\ln{n}$
(otherwise the space bound in Theorem~\ref{thm:general} is enough to store the whole input and solve the problem optimally), the value of
$\cmax$ is provided to the algorithm, and $x$ is a vector of \emph{binary} variables, i.e., $x \in \set{0,1}^{m}$ (hence covering ILP
instances are always referring to covering ILP instances with binary variables); in Section~\ref{sec:remarks}, we describe how to eliminate
the later two assumptions.

\paragraph{Algorithm overview.} For any covering ILP instance $\Ins := \CILP(A,b,c)$, our algorithm
estimates $\opt:= \opt(\Ins)$ in two parts running in parallel. In the first part, the
goal is simply to compute $\cost(\Ins)$ (see Claim~\ref{clm:cost-calculator}). For the
second part, we design a tester algorithm (henceforth, $\testerILP$) that given any
``guess'' $k$ of the value of $\opt$, if $k \ge \opt$, $\testerILP$ accepts $k$ w.p. $1$
and for any $k$ where $\cost(\Ins) \le k \le {\opt \over 32\alpha}$, w.h.p. $\testerILP$
rejects $k$.

Let $K:= \set{2^\gamma}_{\gamma \in [\ceil{\log(m\cmax)}]}$; for each $k \in K$ (in parallel), we run $\testerILP(k)$. At the end of the stream, 
the algorithm knows $\cost(\Ins)$ (using the output of the part one), and hence it can identify among
all guesses that are \emph{at least} $\cost(\Ins)$, the smallest guess accepted by
$\testerILP$ (denoted by $k^*$). On one hand, $k^* \le \opt$ since for any guess
$k \ge \opt$, $k \ge \cost(\Ins)$ also (since $\opt \ge \cost(\Ins)$) and  $\testerILP$ accepts
$k$. On the other hand, $k^* \ge {\opt \over 32\alpha}$ w.h.p. since $(i)$ if
$\cost(\Ins) \ge {\opt \over 32\alpha}$, $k^* \ge \cost(\Ins) \ge {\opt \over 32 \alpha}$
and $(ii)$ if $\cost(\Ins) < {\opt \over 32\alpha}$, the guess ${\opt \over 32 \alpha}$ will be
rejected w.h.p. by $\testerILP$.  Consequently, $32\alpha \cdot k^*$ is an $O(\alpha)$-estimation of
$\opt(\Ins)$.

We first show that one can compute the $\cost$ of a covering ILP presented in a stream using a
simple \emph{dynamic programming} algorithm. 
\begin{claim}\label{clm:cost-calculator}
  For any $\Ins:= \CILP(A,b,c)$ presented by a stream of columns, $\cost(\Ins)$ can be
  computed in space $O(n\bmax\log{\cmax})$ bits.
\end{claim}
\begin{proof}
  We maintain arrays $\cost_j[y] \leftarrow +\infty$ for any $j \in [n]$ and $y \in [\bmax]$: $\cost_j[y]$ is the minimum cost for achieving
  a coverage of $y$ for the $j$-th constraint. Define $\cost_j[y] = 0$ for $y \leq 0$. Upon arrival of $\pair{A_i}{c_i}$, compute $    \cost_j[y] = \min(\cost_j[y],\cost_j[y-a_{j,i}]+c_i)$,
  for every $y$ from $b_j$ down to $1$. We then have $\cost(\Ins)= \max_{j} \cost_j[b_j]$. 
\end{proof}

To continue, we need the following notation. For any vector $v$ with dimension $d$ and any
set $S \subseteq [d]$, $v(S)$ denotes the projection of $v$ onto the dimensions indexed by $S$.  For any two vectors $u$ and $v$, let
$\min(u,v)$ denote a vector $w$ where at the $i$-th dimension: $w_i = \min(u_i, v_i)$, i.e., the \emph{coordinate-wise minimum}.
 We now provide the aforementioned $\testerILP$ algorithm. 

\textbox{$\testerILP(k)$: \textnormal{An algorithm for testing a guess $k$ of the optimal value
    of a covering ILP.}}{
\medskip \\ 
\textbf{Input:} An instance $\Ins := \CILP(A,b,c)$ presented as a stream
$\pair{A_1}{c_1}$, $\ldots$, $\pair{A_m}{c_m}$, a parameter $\alpha \geq 32\ln{n}$, and a guess $k\in K$. \\
\textbf{Output:} \accept
if $k \geq \opt$ and \reject if $\cost(\Ins) \le k \leq {\opt \over 32\alpha}$. The answer could be
  either \accept or \reject if ${\opt \over 32\alpha} < k < \opt$.  \\ 
\algline
\begin{enumerate}
\item \emph{Preprocessing:} 
\begin{enumerate}[(i)] 
\item Maintain an $n$-dimensional vector $\bres \leftarrow b$, an $m$-dimensional vector $\cT \leftarrow 0^{m}$, and an $n \times m$ dimensional matrix $\AT \leftarrow 0^{n \times m}$. 
\item Let $V$ be a subset of $[n]$ obtained by sampling each element in $[n]$ independently with probability $p:= 4\ln{n}/\alpha$.
 \end{enumerate}
\item \emph{Streaming:} when a pair $\pair{A_i}{c_i}$ arrives:
  \begin{enumerate}[(i)]
  \item If $c_i > k$, directly continue to the next input pair of the stream. Otherwise: 
  \item \emph{Prune step}: Let $u_i := \min(\bres, A_i)$ (the coordinate-wise minimum). If $\lOne{u_i} \geq { n \cdot \bmax
    \over \alpha}$, update $\bres \leftarrow \bres - u_i$ (we say $\pair{A_i}{c_i}$ is \emph{pruned} by $\testerILP$ in this case). Otherwise, assign $\AT_i \leftarrow u_i(V)$,
    and $\cT_i \leftarrow c_i$.   \end{enumerate}
\item At the end of the stream, solve the following covering ILP (denoted by $\Ins_{tester}$):
  \begin{align*}
  \text{min~}  \cT \cdot  x
  \text{~~s.t.~~~} \AT x \geq \bres(V)
\end{align*}
  If $\opt(\Ins_{tester})$ is at most $k$, \accept; otherwise \reject.
\end{enumerate} 
}

We first make the following observation. In the prune step of $\testerILP$, if we replace $\AT_i \leftarrow u_i(V)$ by $\AT_i \leftarrow A_i(V)$,
the solution of the resulting covering ILP instance (denoted by $\Ins'_{tester}$) has the property that
$\opt(\Ins'_{tester}) = \opt(\Ins_{tester})$ (we use $\Ins_{tester}$ only to control the space requirement). To see this, let $\bres^i$
denotes the content of the vector $\bres$ when $\pair{A_i}{c_i}$ arrives.  By construction, $(u_i)_j := \min((\bres^i)_j, a_{j,i})$, and
hence if $(u_i)_j \neq a_{j,i}$, then both $(u_i)_j$ and $a_{j,i}$ are at least $(\bres^i)_j$, which is at least $(\bres)_j$ (since every
dimension of $\bres$ is monotonically decreasing).  However, for any \emph{integer} program $\CILP(A,b,c)$, changing any entry $a_{j,i}$ of
$A$ between two values that are at least $b_j$ does not change the optimal value, and hence $\opt(\Ins'_{tester}) = \opt(\Ins_{tester})$. To
simplify the proof, in the following, when concerning $\opt(\Ins_{tester})$, we redefine $\Ins_{tester}$ to be $\Ins'_{tester}$.

We now prove the correctness of $\testerILP$ in the following two lemmas. 

\begin{lemma}\label{lem:ILP-upper}
  For any guess $k \ge \opt$, $\Pr\Paren{\testerILP(k) = \accept} = 1$.
\end{lemma}
\begin{proof}
  Fix any optimal solution $x^*$ of $\Ins$; we will show that $x^*$ is a feasible solution
  for $\Ins_{tester}$, and since by the construction of $\cT$, we have
  $\cT \cdot x^* \le c \cdot x^* \le \opt$, this will show that
  $\opt(\Ins_{tester}) \leq \opt \leq k$ and hence $\testerILP(k) = \accept$.

  Fix a constraint $j$ in $\Ins_{tester}$ as follows: 
  \[
    \sum_{i \in [m]} \aT_{j,i} x_i \ge \bres(V)_j
  \]
  If $j \notin V$, $\bres(V)_j =0$ and the constraint is trivially satisfied for any solution $x^*$.  Suppose $j \in V$ and let $P$ denote the set of (indices of) pairs
  that are pruned.  By construction of the $\testerILP$, $\bres(V)_j = \max(b_j - \sum_{i \in P} a_{j,i}, 0)$. If $\bres(V)_j = 0$, again the
  constraint is trivially satisfied. Suppose $\bres(V)_j = b_j - \sum_{i \in P} a_{j,i}$. The constraint $j$ can be written as
  \[
    \sum_{i} \aT_{j,i} x_i \ge b_j - \sum_{i \in P} a_{j,i}
  \]
  By construction of the tester, $\aT_{j,i} = 0$ for all $i$ that are pruned and otherwise $\aT_{j,i} = a_{j,i}$. Hence, we can further
  write the constraint $j$ as
  \[
    \sum_{i \notin P} a_{j,i} x_i \ge b_j - \sum_{i \in P} a_{j,i}
  \]
  Now, since $x^*$ satisfies the constraint $j$ in $\Ins$,
  \begin{align*}
        \sum_{i \in [m]} a_{j,i} x^*_i &\ge b_j \\ 
    \sum_{i \notin P} a_{j,i} x^* &\ge b_j - \sum_{i \in P} a_{j,i} x^*_i \\
    &\ge b_j - \sum_{i \in P} a_{j,i}  \tag{$x^*_i \le 1$}
  \end{align*}
  and the constraint $j$ is satisfied.  Therefore, $x^*$ is a feasible solution of
  $\Ins_{tester}$; this completes the proof.
\end{proof}

We now show that $\testerILP$ will reject guesses that are smaller than
${\opt \over 32\alpha}$. We will only prove that the rejection happens with probability
$3/4$; however, the probability of error can be reduced to any $\delta < 1$ by running $O(\log{1/\delta})$ parallel
instances of the $\testerILP$ and for each guess, \reject if any one of the instances
outputs \reject and otherwise \accept. In our case $\delta = O(\card{K}^{-1})$ so we can apply union bound for all different guesses. 

\begin{lemma}\label{lem:ILP-lower}
  For any guess $k$ where $\cost(\Ins) \le k < {\opt \over 32\alpha}$,
  $\Pr\Paren{\testerILP(k) = \reject} \geq 3/4$.
\end{lemma}
\begin{proof}
  By construction of $\testerILP(k)$, we need to prove that $\Pr\Paren{\opt(\Ins_{tester}) > k} \geq 3/4$.  Define the following
  covering ILP $\Ins'$:
    \begin{align*}
  \text{min~}  \cT \cdot  x
  \text{~~s.t.~~~} \AH x \geq \bres(V)
\end{align*}
where $\AH_i = A_i$ if $\pair{A_i}{c_i}$ is not pruned by \testerILP, and $\AH_i = 0^n$ otherwise. In $\testerILP(k)$, for each pair
$\pair{A_i}{c_i}$ that is not pruned, instead of storing the entire vector $A_i$, we store the projection of $A_i$ onto dimensions indexed
by $V$ (which is the definition of $\AT_i$ in $\Ins_{tester}$).  This is equivalent to performing constraint sampling on $\Ins'$ with a
sampling rate of $p = 4\ln{n}/\alpha$.  Therefore, by Lemma~\ref{lem:sub-sample-save-opt}, with probability at least $3/4$,
$\opt(\Ins_{tester}) + \cost(\Ins') \ge {\opt(\Ins') \over 8\alpha}$.  Since $\cost(\Ins') \le \cost(\Ins) \le k < {\opt(\Ins) \over 32\alpha}$, this
implies that
\begin{align*}
\opt(\Ins_{tester}) \ge {\opt(\Ins') \over 8\alpha} - \cost(\Ins') > {\opt(\Ins') \over 8\alpha} - {\opt(\Ins) \over 32\alpha}.  
\end{align*}
Therefore, we only need to show that $\opt(\Ins') \ge {\opt(\Ins) \over 2}$ since then
$\opt(\Ins_{tester}) > {\opt(\Ins) \over 16\alpha}  - {\opt(\Ins) \over 32\alpha} = {\opt(\Ins) \over 32\alpha} > k$ and $\testerILP$ will reject $k$.

To show that $\opt(\Ins') \ge {\opt(\Ins) \over 2}$, we first note that for any optimal
solution $x^*$ of $\Ins'$, if we further set $x^*_i = 1$ for any pair $\pair{A_i}{c_i}$
that are pruned, the resulting $x^*_i$ is a feasible solution for $\Ins$. Therefore, if we
show that the total weight of the $\pair{A_i}{c_i}$ pairs that are pruned is
at most ${\opt \over 2}$, $\opt(\Ins')$ must be at least ${\opt \over 2}$ or we will have
a solution for $\Ins$ better than $\opt(\Ins)$.

To see that the total weight of the pruned pairs is at most $\opt/2$, since only pairs
with $c_i \le k$ ($\le {\opt \over 32\alpha}$) will be considered, we only need to show
that at most $16\alpha$ pairs can be pruned. By the construction of the prune step, each
pruned pair reduces the $\ell_1$-norm of the vector $b_{res}$ by an additive factor of at
least ${n \bmax \over \alpha}$. Since $b_{res}$ is initialized to be $b$ and
$\lOne{b} \le n\bmax$, at most $\alpha$ ($\le 16\alpha$) pairs can be pruned.  This
completes the proof.
\end{proof}

We now finalize the proof of Theorem~\ref{thm:general}.

\begin{proof}[Proof of Theorem~\ref{thm:general}]
We run the algorithm described in the beginning of this section. The correctness of the
algorithm follows from Claim~\ref{clm:cost-calculator} and Lemmas~\ref{lem:ILP-upper}
and~\ref{lem:ILP-lower}.  We now analyze the space complexity of this algorithm. We need
to run the algorithm in Claim~\ref{clm:cost-calculator} to compute $\cost(\Ins)$, which
require $\Ot(n\bmax)$ space.  We also need to run $\testerILP$ for $O(\log{(m\cdot\cmax)})$
different guesses of $k$. 

In $\testerILP(k)$, we need $O(n\log{\bmax})$ bits to store the
vector $\bres$ and $O(m\log{\cmax})$ bits to maintain the vector $\cT$. Finally, the matrix
$\AT$ requires $O(mn\bmax/\alpha \cdot (\log{n}/\alpha) \cdot (\log{\amax}\log{n}))$ bits
to store. This is because each column $\AT_i$ of $\AT$ is either $0^n$ or $u_i(V)$ where
$\lOne{u_i} < {n\cdot \bmax \over \alpha}$. Since
$\lOne{u_i} < {n\cdot \bmax \over \alpha}$, there are at most
${n\cdot \bmax \over \alpha}$ non-zero entries in $u_i$. Therefore, after projecting $u_i$
to $V$ (to obtain $\AT_i$) in expectation the number of non-zero entries in $\AT_i$ is at most $\Ot(n\bmax/\alpha^2)$. 
Using the Chernoff bound w.h.p at most $O({n\bmax \over \alpha^2})$ non-zero entries of $u_i$
remain in each $\AT_i$, where each entry needs $O(\log{\amax}\log{n})$ bits to store. Note that the space complexity of the algorithm can be made \emph{deterministic} by simply terminating the execution when at least one 
set $\AT_i$ has $(c \cdot \frac{n\bmax}{\alpha^2})$ non-zero entries (for a sufficiently large constant $c > 1$); 
as this event happens with $o(1)$ probability, the error probability of the algorithm increases only by $o(1)$.  Finally, as all entries
in $(A,b,c)$ are $\poly(n)$-bounded, the total space requirement of the algorithm is
$\Ot((mn/\alpha^2)\cdot\bmax + m + n\bmax)$.
\end{proof}

We also make the following remark about $\alpha$-approximating covering ILPs. 

\begin{remark}
  The simple algorithm described in Section~\ref{sec:techniques} for
  $\alpha$-approximating set cover can also be extended to obtain an
  \emph{$\alpha$-approximation algorithm} for covering ILPs in space $\Ot(mn\bmax/\alpha)$: Group
  the columns by the weights and merge every $\alpha$ sets for each group independently.
\end{remark}

\subsection{Further Remarks}\label{sec:remarks}
We now briefly describe how to eliminate the assumptions that variables are binary and $\cmax$
is given to the algorithm.

\ssection{Binary variables versus integer variables.} We point out that any algorithm
that solves covering ILP with \emph{binary} variables in the streaming setting, can also
be used to solve covering ILP with \emph{non-negative integer} variables (or shortly,
integer variables), while increasing the number of columns by a factor of $O(\log \bmax)$.

To see this, given any covering ILP instance $\CILP(A,b,c)$ (with integer variables), we
replace each $\pair{A_i}{c_i}$ pair with $O(\log{\bmax})$ pairs (or columns) where the
$j$-th pair ($j \in [\ceil{\log{\bmax}}]$) is defined to be
$\pair{2^{j-1} A_i}{2^{j-1}c_i}$. Every combination of the corresponding $j$ binary
variables maps to a unique binary representation of the variable $x_i$ in $\CILP(A,b,c)$
for $x_i = O(\bmax)$.  Since no variable in $\CILP(A,b,c)$ needs to be more than $\bmax$,
the correctness of this reduction follows.

\ssection{Knowing $\cmax$ versus not knowing $\cmax$.} As we pointed out earlier, the
assumption that the value of $\cmax$ is provided to the algorithm is only for simplicity;
here we briefly describe how to eliminate this assumption. Our algorithm uses $\cmax$ to
determine the range of $\opt$, which determines the set of guesses $K$ that $\testerILP$
needs to run. If $\cmax$ is not provided, one natural fix is to use the same randomness
(i.e., the same set $V$) for all testers, and whenever a larger $c_i$ arrives, create a
tester for each new guess by duplicating the state of $\testerILP$ for the current largest
guess and continue from there (the correctness is provided by the design of our
$\testerILP$).

However, if we use this approach, since we do not know the total number of guesses
upfront, we cannot boost the probability of success to \emph{ensure} that w.h.p. no tester
fails (if $\cmax$ is $\poly(n)$-bounded, the number of guesses is
$O(\log(n\cmax)) = O(\log{n})$, $O(\log\log{n})$ parallel copies suffice, and though we do
not know the value of the constant, $\log{n}$ copies always suffice; but this will not be
the case for general $\cmax$). To address this issue, we point out that it suffices to
ensure that the copy of $\testerILP$ that runs a specific guess $k'$ succeed w.h.p., where
$k'$ is the largest power of $2$ with $k' \le {\opt \over 32\alpha}$. To see this, we
change the algorithm to pick the largest rejected guess $k^*$ and return $64\alpha k^*$
(previously, it picks the smallest accepted guess $k$ and returns $32 \alpha k$), if $k'$
is correctly rejected by the tester, $k^* \ge k' \ge {\opt \over 64 \alpha}$.  On the
other hand, for any guess no less than $\opt$, $\testerILP$ accepts it w.p. $1$, and hence
$k^* \le \opt$.  Therefore, as long as $\testerILP$ rejects $k'$ correctly, the output of
the algorithm is always an $O(\alpha)$-estimation and hence we do not need the knowledge of $\cmax$ upfront.

\ssection{The set multi-cover problem.}  For the set multi-cover problem,
Theorem~\ref{thm:general} gives an $\alpha$-estimation algorithm using space
$\Ot({mn \bmax \over \alpha^2} + m + n \bmax)$; here, $\bmax$ is the maximum \emph{demand}
of the elements.  However, we remark that if sets are unweighted, a better space requirement of
$\Ot({mn \over \alpha^2} + m + n)$ is achievable for this problem: Instead of running the
tester for multiple guesses, just run it once with the following changes.  In the prune
step, change from ``$\lOne{u_i} \ge {n \bmax \over \alpha}$'' to
``$\lOne{u_i} \ge {n \over \alpha}$''; and output
$est := \#_{pruned} + 8\alpha (\opt(\Ins_{tester}) + \bmax)$, where $\#_{pruned}$ is the
number of sets that are pruned. One one hand, $est \ge opt$ w.h.p since $\#_{pruned}$ plus the optimal
solution of the residual covering ILP (denoted by $\opt_{res}$) is a feasible solution of
$\Ins$, and by Lemma~\ref{lem:sub-sample-save-opt}, w.h.p.
$8\alpha (\opt(\Ins_{tester}) + \bmax) \ge \opt_{res}$.  One the other hand, $est = O(\alpha \cdot \opt )$ since
$\#_{pruned} \le \alpha \bmax \le \alpha \cdot \opt$ ($\bmax$ sets is needed even for covering one
element); $\opt(\Ins_{tester}) \le \opt$ and $\bmax \le \opt$, which implies
$8\alpha (\opt(\Ins_{tester}) + \bmax) \le 8\alpha(\opt + \opt) = O(\alpha \cdot \opt)$.

\section{An $\Omega(mn/\alpha^2)$-Space Lower Bound for $\alpha$-Estimate Set Cover}\label{sec:rand-lower}

Our algorithm in Theorem~\ref{thm:general} suggests that there exists at least a factor $\alpha$ gap on the space requirement of $\alpha$-approximation and $\alpha$-estimation algorithms for the
set cover problem. We now show that this gap is the best possible. In other words, the space complexity of our algorithm in Theorem~\ref{thm:general} for the original set cover problem is tight (up to logarithmic factors) 
\emph{even for random arrival streams}. Formally,

\begin{theorem}\label{thm:rand-lower}
  Let $\SS$ be a collection of $m$ subsets of $[n]$ presented one by one in a \emph{random
    order}. For any $\alpha = o(\sqrt{\frac{n}{\log{n}}})$ and any $m = \poly(n)$, any
  \emph{randomized} algorithm that makes a single pass over $\SS$ and outputs an
  $\alpha$-estimation of the set cover problem with probability $0.9$ (over the randomness of both the
  stream order and the algorithm) must use ${\Omgt}(\frac{mn}{\alpha^2})$ bits of space.
\end{theorem}


Fix a (sufficiently large) value for $n$, $m = \poly(n)$, and $\alpha = o(\sqrt{\frac{n}{\log{n}}})$;
throughout this section, \setCoverEst refers to the problem of $\alpha$-estimating the set cover problem with $m+1$ sets (see footnote~\ref{footnote:m+1}) defined over the universe $[n]$ in the one-way communication model, whereby the sets are partitioned between Alice and Bob.

\paragraph{Overview.} We start by introducing a hard distribution $\diste$ for \setCoverEst in the spirit of the distribution $\dista$ in Section~\ref{sec:lower-find}. However, 
since in \setCoverEst the goal is only to estimate the \emph{size} of the optimal cover, ``hiding'' one single element (as was done in $\dista$) is not enough for the lower bound. Here, instead of
trying to hide a single element, we give Bob a ``block'' of elements and his goal would be to decide whether this block appeared in a single set of Alice as a whole or was it 
partitioned across many different sets\footnote{The actual set given to Bob is the complement of this block; hence the optimal set cover size varies significantly between the two cases.}. 
Similar to $\dista$, distribution $\diste$ is also not a product distribution; however, we introduce a way of decomposing $\diste$ into a convex combination
of \emph{product distributions} and then exploit the simplicity of product distributions to prove the lower bound. 

Nevertheless, the distribution $\diste$ is still ``adversarial'' and hence is not suitable for proving the lower bound for random arrival streams. Therefore, we define
an extension to the original hard distribution as $\dister$ which \emph{randomly} partitions the sets of distribution $\diste$ between Alice and Bob. We prove a lower 
bound for this distribution using a reduction from protocols over $\diste$. Finally, we show how an algorithm for set cover over random arrival streams would be able to solve
instances of \setCoverEst over $\dister$ and establish Theorem~\ref{thm:rand-lower}.

\subsection{A Hard Input Distribution for \setCoverEst} \label{sec:sc-hard-dist}

Consider the following distribution $\diste$ for \setCoverEst. 

\textbox{Distribution $\diste$. \textnormal{A hard input distribution for \setCoverEst.}}{
	\medskip \\ 
        \textbf{Notation.} Let $\FC$ be the collection of all subsets of $[n]$ with
        cardinality $\frac{n}{10\alpha}$.
\begin{itemize}
\item \textbf{Alice.} The input of Alice is a collection of $m$ sets $\SA =
  \paren{S_1,\ldots,S_m}$, where for any $i \in [m]$, $S_i$ is a set chosen independently and uniformly at random from $\FC$.
\item \textbf{Bob.} Pick $\theta \in \set{0,1}$ and $\istar \in [m]$ independently and
  uniformly at random; the input of Bob is a single set $T$ defined as follows.
  \begin{itemize}
  \item If $\theta = 0$, then $\bar{T}$ is a set of size $\alpha \log{m}$ chosen uniformly
    at random from $S_\istar$.\footnote{
    Since $\alpha = o(\sqrt{n/\log{n}})$ and $m = \poly(n)$, the size of
    $\bar{T}$ is strictly smaller than the size of $S_\istar$.}
  \item If $\theta = 1$, then $\bar{T}$ is a set of size $\alpha \log{m}$ chosen uniformly
    at random from $[n] \setminus S_\istar$.
  \end{itemize}
\end{itemize}
}

Recall that $\opt(\SA,T)$ denotes the \emph{set cover size} of the
input instance $(\SA,T)$. We first establish the following lemma regarding the parameter $\theta$ and $\opt(\SA,T)$
in the distribution $\diste$.

\begin{lemma}\label{lem:theta-random}
	For $(\SA,T) \sim \diste$: 
	\begin{enumerate}[(i)]
		\item \label{part:theta0} $\PR{\opt(\SA,T) = 2 \mid \theta=0} = 1$. 
		\item \label{part:theta1} $\PR{\opt(\SA,T) > 2\alpha \mid \theta = 1} = 1-o(1)$. 
	\end{enumerate}
\end{lemma}
\begin{proof}
  Part (\ref{part:theta0}) is immediate since by construction, when $\theta = 0$, $T
  \union S_\istar = [n]$. We now prove part (\ref{part:theta1}).
	
  Since a valid set cover must cover $\bar{T}$, it suffices for us to show that w.h.p. no
  $2\alpha$ sets from $\SA$ can cover $\bar{T}$. By construction, neither $T$ nor
  $S_\istar$ contains any element in $\bar{T}$, hence $\bar{T}$ must be covered by at most
  $2\alpha$ sets in $\SA \setminus \set{S_\istar}$. 
	
  Fix a collection $\Salpha$ of $2\alpha$ sets in $\SA \setminus \set{S_\istar}$; we first
  analyze the probability that $\Salpha$ covers $\bar{T}$ and then take union bound over
  all choices of $2\alpha$ sets from $\SA \setminus \set{S_\istar}$.  Note that according to the distribution $\diste$, the sets in
  $\Salpha$ are drawn independent of $\bar{T}$. Fix any choice of $\bar{T}$; for each
  element $k \in \bar{T}$, and for each set $S_j \in\Salpha$, define an indicator random
  variable $\bX^{j}_k \in \set{0,1}$, where $\bX^{j}_k=1$ iff $k \in S_j$.  Let $\bX :=
  \sum_{j}\sum_{k} \bX^{j}_k$ and notice that:
	\[
	\Ex[\bX] = \sum_{j}\sum_{k} \Ex[\bX^{j}_k] = (2\alpha) \cdot (\alpha\log{m}) \cdot (\frac{1}{10\alpha}) = \alpha\log{m}/5
	\]
	We have,
	\[
		\PR{\text{$\Salpha$ covers $\bar{T}$}} \leq \PR{\bX \geq \alpha\log{m}} = \PR{\bX \geq 5\Ex[\bX]} \leq \exp\paren{-3\alpha\log{m}}
	\]
	where the last equality uses the fact that $\bX^{j}_k$ variables are negatively
        correlated (which can be proven analogous to Lemma~\ref{lem:cover-size}) and applies the extended Chernoff bound (see Section~\ref{sec:prelim}).  Finally, by union bound,
	\begin{align*}
		\Pr(\opt(\SA,T) \leq 2\alpha) &\leq \PR{\exists~\Salpha \text{ covers } \bar{T}} \leq {m \choose 2\alpha} \cdot \exp\paren{-3\alpha\log{m}} \\
		&\leq \exp\paren{2\alpha\cdot\log{m}-3\alpha\log{m}}= o(1)
	\end{align*}
\end{proof}

Notice that distribution $\diste$ is not a product distribution due to the correlation
between the input given to Alice and Bob. However, we can express the distribution as a
convex combination of a relatively small set of product distributions; this significantly simplifies the 
proof of the lower bound. To do so, we need the following definition.  For
integers $k,t$ and $n$, a collection $P$ of $t$ subsets of $[n]$ is called a \emph{random
  $(k,t)$-partition} iff the $t$ sets in $P$ are constructed as follows: Pick $k$ elements
from $[n]$, denoted by $S$, uniformly at random, and partition $S$ randomly into $t$ sets
of equal size.  We refer to each set in $P$ as a \emph{block}.

\textbox{An alternative definition of the distribution \diste.}{
\begin{enumerate}
	\item[] \textbf{Parameters:} $~~~~~~~~~~k = \frac{n}{5\alpha}~~~~~~~~~~p = \alpha\log{m}~~~~~~~~~~t = k/p$
	\item For any $i \in [m]$, let $P_i$ be a random $(k,t)$-partition in $[n]$ (chosen independently). 
	\item The input to Alice is $\SA = \paren{S_1,\ldots,S_m}$, where each $S_i$ is created by picking $t/2$ blocks from $P_i$ uniformly at random. 
	\item The input to Bob is a set $T$ where $\bar{T}$ is created by first picking an
          $\istar \in [m]$ uniformly at random, and then picking a block from $P_\istar$
          uniformly at random.
\end{enumerate}
}
To see that the two formulations of the distribution $\diste$ are indeed equivalent, notice that $(i)$ the input given to Alice in the new formulation is a collection of sets of size $n/10\alpha$ chosen 
independently and uniformly at random (by the independence of $P_i$'s), and $(ii)$ the complement of the set given to Bob is a set of size $\alpha\log{m}$ which, for $\istar \in_R [m]$, with probability half, is chosen uniformly at 
random from $S_\istar$, and with probability half, is chosen from $[n] \setminus S_{\istar}$ (by the randomness in the choice of each block in $P_\istar$). 


\newcommand{\bTbar}{\ensuremath{\mathbf{\bar{T}}}\xspace}

Fix any $\delta$-error protocol $\protsc$ ($\delta < 1/2$) for \setCoverEst on the
distribution $\diste$. Recall that $\bprotsc$ denotes the random variable for the concatenation of the message of Alice with the 
public randomness used in the protocol $\protsc$. We further use $\bPP := (\bP_1,\ldots,\bP_t)$ to denote the random partitions $(P_1,\ldots,P_{t})$, $\bI$ for the choice of the special index $\istar$, and $\bT$ 
for the parameter $\theta \in \set{0,1}$, whereby $\theta = 0$ iff $\bar{T} \subseteq S_\istar$.


We make the following simple observations about the distribution $\diste$. The proofs are straightforward. 

\begin{remark}\label{rem:dist} In the distribution $\diste$, 
	\begin{enumerate}
		\item \label{p1} The random variables $\bSS$, $\bPP$, and $\bprotsc(\bSS)$ are all
                  independent of the random variable $\bI$.
		\item \label{p2} For any $i \in [m]$, conditioned on $\bP_i = P$, and $\bI = i$, the
                  random variables $\bS_i$ and $\bTbar$ are independent of each
                  other. Moreover, $\supp{\bS_i}$ and $\supp{\bTbar}$ contain,
                  respectively, ${t \choose \frac{t}{2}}$ and $t$ elements and both
                  $\bS_i$ and $\bTbar$ are uniform over their support.
		\item \label{p3} For any $i \in [m]$, the random variable $\bS_i$ is
                  independent of both $\bSS^{-i}$ and $\bPP^{-i}$. 
	\end{enumerate}
\end{remark}

\subsection{The Lower Bound for the Distribution $\diste$}\label{sec:rand-lb}

Our goal in this section is to lower bound $\ICost{\protsc}{\diste}$ and ultimately $\norm{\protsc}$. We start by simplifying the expression 
for $\ICost{\protsc}{\diste}$. 

\begin{lemma}\label{lem:info-cost-bound}
	$\ICost{\protsc}{\diste} \geq \sum_{i=1}^{m} I(\bprotsc ; \bS_i \mid \bP_i)$
\end{lemma}
\begin{proof}
	We have,
	\begin{align*}
		\ICost{\protsc}{\diste}&= I(\bprotsc; \bSS) \geq I(\bprotsc ; \bSS \mid \bPP) 
	\end{align*}
	where the inequality holds since $(i)$ $H(\bprotsc) \geq H(\bprotsc \mid \bPP)$ and $(ii)$ $H(\bprotsc \mid \bSS) = H(\bprotsc \mid \bSS,\bPP)$ as $\bprotsc$ is independent of $\bPP$ conditioned on $\bSS$.
	We now bound the conditional mutual information term in the above equation.  
	\begin{align*}
		 I(\bprotsc ; \bSS \mid \bPP) &= \sum_{i=1}^{m} I(\bS_i ; \bprotsc \mid \bPP,\bSS^{<i}) \tag{the chain rule for the mutual information, Claim~\ref{clm:it-facts}-(\ref{part:chain-rule})} \\
		&= \sum_{i=1}^{m} H(\bS_i \mid \bPP,\bSS^{<i}) - H(\bS_i \mid \bprotsc,\bPP,\bSS^{<i}) \\
		&\geq \sum_{i=1}^{m} H(\bS_i \mid \bP_i) - H(\bS_i \mid \bprotsc,  \bP_i) \\
                 & = \sum_{i=1}^{m} I(\bS_i ; \bprotsc \mid \bP_i)
	\end{align*}
	The inequality holds since:
	\begin{enumerate}[(i)]
	\item $H(\bS_i \mid \bP_i) = H(\bS_i \mid \bP_i,\bPP^{-i}, \bSS^{<i}) = H(\bS_i \mid \bPP,\bSS^{<i})$ 
	because conditioned on $\bP_i$, $\bS_i$ is independent of $\bPP^{-i}$ and $\bSS^{<i}$ (Remark~\ref{rem:dist}-(\ref{p3})), hence the equality holds by Claim~\ref{clm:it-facts}-(\ref{part:cond-reduce}).
	\item
        $H(\bS_i \mid \bprotsc,\bP_i) \geq H(\bS_i \mid \bprotsc, \bP_i, \bPP^{-i},\bSS^{<i}) =
        H(\bS_i \mid \bprotsc, \bPP,\bSS^{<i})$ since conditioning reduces the entropy, i.e.,
        Claim~\ref{clm:it-facts}-(\ref{part:cond-reduce}).
	\end{enumerate}	
\end{proof}

Equipped with Lemma~\ref{lem:info-cost-bound}, we only need to bound  $\sum_{i \in [m]} I(\bprotsc ; \bS_i \mid \bP_i)$. Note that, 
\begin{align}
	\sum_{i=1}^{m} I(\bprotsc ; \bS_i \mid \bP_i) = \sum_{i=1}^{m} H(\bS_i \mid \bP_i) - \sum_{i=1}^{m} H(\bS_i \mid \bprotsc , \bP_i) \label{eq:minfo-bound}
\end{align}
Furthermore, for each $i \in [m]$, $\card{\supp{\bS_i \mid \bP_i}} = {t \choose \frac{t}{2}}$ and $\bS_i$ is uniform over its support (Remark~\ref{rem:dist}-(\ref{p2})); hence, by Claim~\ref{clm:it-facts}-(\ref{part:uniform}),
\begin{align}
	\sum_{i=1}^{m} H(\bS_i \mid  \bP_i) = \sum_{i=1}^{m} \log{t \choose \frac{t}{2}} = m\cdot \log{\paren{2^{t - \Theta(\log{t})}}} = m\cdot t - \Theta(m\log{t}) \label{eq:ent-bound}
\end{align} 

Consequently, we only need to bound $\sum_{i=1}^{m} H(\bS_i \mid \bprotsc, \bP_i)$. In order to do so, we show that $\protsc$ can be used to estimate the value of
 the parameter $\theta$, and hence we only need to establish a lower bound for the problem of estimating $\theta$.

\begin{lemma}\label{lem:sc-to-theta}
  Any $\delta$-error protocol $\protsc$ over the distribution $\diste$ can be used to
  determine the value of $\theta$ with error probability $\delta + o(1)$.
\end{lemma} 
\begin{proof}
  Alice sends the message $\protsc(\SS)$ as before. Using this message, Bob can compute an
  $\alpha$-estimation of the set cover problem using $\protsc(\SS)$ and his input. 
  If the estimation is less than $2\alpha$, we output $\theta = 0$ and otherwise we output $\theta=1$. The
  bound on the error probability follows from Lemma~\ref{lem:theta-random}.
\end{proof}

Before continuing, we make the following remark which would be useful in the next section. 

\begin{remark}\label{rem:known-i}
  We assume that in \setCoverEst over the distribution $\diste$, Bob is additionally provided with the special index $\istar$. 
\end{remark}

Note that this assumption can only make our lower bound stronger since Bob can always ignore this information and solves the original \setCoverEst. 

Let $\beta$ be the function that estimates $\theta$ used in Lemma~\ref{lem:sc-to-theta}; the input to $\beta$ is the message given from Alice, the public coins used by the players, 
the set $\bar{T}$, and (by Remark~\ref{rem:known-i}) the special index $\istar$. We have,
\begin{align*}
  \Pr(\beta(\bprotsc,\bTbar,\bI) \neq \bT) \leq \delta+o(1)
\end{align*}

Hence, by Fano's inequality (Claim~\ref{clm:fano}),
\begin{align}
	H_2(\delta+o(1)) &\geq H(\bT \mid \bprotsc,\bTbar, \bI) \notag \\
	&= \Ex_{i \sim \bI} \Bracket{H(\bT \mid \bprotsc,\bTbar, \bI = i)} \notag \\
	&= \frac{1}{m} \sum_{i=1}^{m} H(\bT \mid \bprotsc,\bTbar, \bI = i) \label{eq:i-sum}
\end{align}

We now show that each term above is lower bounded by $H(\bS_i \mid \bprotsc,\bP_i) / t$ and hence we obtain the desired upper bound on $H(\bS_i \mid \bprotsc,\bP_i)$ in Equation~(\ref{eq:minfo-bound}).

\begin{lemma}\label{lem:theta-to-s}
  For any $i \in [m]$, $H(\bT \mid \bprotsc , \bTbar , \bI = i) \geq  H(\bS_i \mid \bprotsc,\bP_i) / t$. 
\end{lemma}
\begin{proof}
	We have,
	\begin{align*}
		H(\bT \mid \bprotsc, \bTbar , \bI = i)&\geq H(\bT \mid \bprotsc,\bTbar ,  \bP_i , \bI = i) \tag{conditioning on random variables reduces entropy, Claim~\ref{clm:it-facts}-(\ref{part:cond-reduce})} \\
		&= \Ex_{P \sim \bP_i \mid \bI = i} \Bracket{H(\bT \mid \bprotsc,\bTbar,  \bP_i = P, \bI = i)}
	\end{align*}
	For brevity, let $E$ denote the event $(\bP_i = P , \bI = i)$.  We can write the above equation as,
	\begin{align*}
		 H(\bT \mid \bprotsc,\bTbar ,  \bP_i , \bI = i) = \Ex_{P \sim \bP_i \mid \bI = i} \Ex_{\bar{T} \sim \bTbar \mid E} \Bracket{H(\bT \mid \bprotsc,\bTbar = \bar{T},  E)}
	\end{align*}
	Note that by Remark~\ref{rem:dist}-(\ref{p2}), conditioned on the event $E$,
        $\bar{T}$ is chosen to be one of the blocks of $P = (B_1,\ldots,B_t)$ uniformly at
        random. Hence,
	\begin{align*}
          H(\bT \mid \bprotsc,\bTbar ,  \bP_i , \bI = i) &= \Ex_{P \sim \bP_i \mid \bI = i} \Bracket{ \sum_{j=1}^{t}  \frac{H(\bT \mid \bprotsc,\bTbar = B_j, E )}{t}}
	\end{align*}
             
        Define a random variable $\bX:= (\bX_1,\ldots,\bX_{t})$, where each $\bX_j \in \set{0,1}$ and 
        $\bX_j = 1$ iff $\bS_i$ contains the block $B_j$. Note that conditioned on $E$,
        $\bX$ uniquely determines the set $\bS_i$. Moreover, notice that
        conditioned on $\bTbar = B_j$ and $E$, $\bT = 0$ iff $\bX_j = 1$. Hence,
	\begin{align*}
		 H(\bT \mid \bprotsc,\bTbar ,  \bP_i , \bI = i) &= \Ex_{P \sim \bP_i \mid \bI = i} \Bracket{\sum_{j=1}^{t} \frac{H(\bX_j \mid \bprotsc,\bTbar = B_j, E )}{t}} 
	\end{align*}
	Now notice that $\bX_j$ is independent of the event $\bTbar = B_j$ since $\bS_i$
        is chosen independent of $\bTbar$ conditioned on $E$
        (Remark~\ref{rem:dist}-(\ref{p2})). Similarly, since $\bprotsc$ is only a
        function of $\bS$ and $\bS$ is independent of $\bTbar$ conditioned on $E$, $\bprotsc$
        is also independent of the event $\bTbar = B_j$. Consequently, by
        Claim~\ref{clm:it-facts}-(\ref{part:ent-event}), we can ``drop'' the conditioning
        on $\bTbar = B_j$,
	\begin{align*}
	H(\bT \mid \bprotsc,\bTbar ,  \bP_i , \bI = i) &= \Ex_{P \sim \bP_i \mid \bI = i} \Bracket{ \sum_{j=1}^{t} \frac{H(\bX_j \mid \bprotsc, E )}{t}}  \\ 
		 &\geq \Ex_{P \sim \bP_i \mid \bI = i} \Bracket{ \frac{H(\bX \mid \bprotsc, E )}{t}}  \tag{sub-additivity of the entropy, Claim~\ref{clm:it-facts}-(\ref{part:sub-additivity})} \\
		 &= \Ex_{P \sim \bP_i \mid \bI = i} \Bracket{\frac{H(\bS_i \mid \bprotsc, E)}{t}} \tag{$\bS_i$ and $\bX$ uniquely define each other conditioned on $E$} \\
		 &= \Ex_{P \sim \bP_i \mid \bI = i} \Bracket{\frac{H(\bS_i \mid \bprotsc, \bP_i = P, \bI = i)}{t}} \tag{$E$ is defined as $(\bP_i = P, \bI = i)$} \\
		 &= \frac{H(\bS_i \mid \bprotsc,\bP_i,\bI = i)}{t}
	\end{align*}
	Finally, by Remark~\ref{rem:dist}-(\ref{p1}), $\bS_i$, $\bprotsc$, and $\bP_i$ are all independent of the event $\bI = i$, and hence by Claim~\ref{clm:it-facts}-(\ref{part:ent-event}), 
	$H(\bS_i \mid  \bprotsc,\bP_i ,\bI = i) = H(\bS_i \mid \bprotsc,\bP_i)$, which concludes the proof.
\end{proof}

By plugging in the bound from Lemma~\ref{lem:theta-to-s} in Equation~(\ref{eq:i-sum}) we have, 
\begin{align*}
	\sum_{i=1}^{m} H(\bS_i \mid \bprotsc , \bP_i) \leq H_2(\delta+o(1)) \cdot (mt)
\end{align*}
Finally, by plugging in this bound together with the bound from Equation~(\ref{eq:ent-bound}) in Equation~(\ref{eq:minfo-bound}), we get,
\begin{align*}
		\sum_{i=1}^{m} I(\bprotsc ; \bS_i \mid \bP_i) &\geq mt - \Theta(m\log{t}) - H_2(\delta+o(1)) \cdot (mt) \\
		&= \Paren{1-H_2(\delta+o(1))} \cdot (mt) - \Theta(m \log{t})
\end{align*}

Recall that $t = k/p = \Omgt(n/\alpha^2)$ and $\delta < 1/2$, hence $H_2(\delta + o(1)) = 1-\eps$ for some constant $\eps$ bounded away from $0$. By Lemma~\ref{lem:info-cost-bound},

\begin{align*}
	\IC{\setCoverEst}{\diste}{\delta} = \min_{\protsc}\Paren{\ICost{\protsc}{\diste}} = \Omgt(mn/\alpha^2)
\end{align*}
To conclude, since the information complexity is a lower bound on the communication complexity (Proposition~\ref{prop:cc-ic}), we obtain the following theorem. 

\begin{theorem}\label{thm:two-adv}
	For any constant $\delta < 1/2$, any $\alpha = o(\sqrt{n/\log{n}})$, and any $m = \poly(n)$, 
	\begin{align*}
		\CC{\setCoverEst}{\diste}{\delta} = \Omgt(mn/\alpha^2)
	\end{align*}
\end{theorem}

As a corollary of this result, we have that the space complexity of single-pass streaming algorithms for the set cover
problem on \emph{adversarial streams} is $\Omgt(mn/\alpha^2)$.

\subsection{Extension to Random Arrival Streams}\label{sec:rand-extension}
We now show that the lower bound established in Theorem~\ref{thm:two-adv} can be further strengthened to prove a
lower bound on the space complexity of single-pass streaming algorithms in the random
arrival model. To do so, we first define an extension of the distribution $\diste$, denoted
by $\dister$, prove a lower bound for $\dister$, and then show that how to use this
lower bound on the one-way communication complexity to establish a lower bound for the random
arrival model.

We define the distribution $\dister$ as follows. 
\textbox{Distribution
  $\dister$. \textnormal{An extension of the hard distribution $\diste$ for \setCoverEst}.}{
  \begin{enumerate}
  \item Sample the sets $\SS = \set{S_1,\ldots,S_m,T}$ in the same way as in the
    distribution $\diste$.
  \item Assign each set in $\SS$ to Alice with probability $1/2$, and the remaining sets
    are assigned to Bob.
  \end{enumerate}
}

We prove that the distribution $\dister$ is still a hard distribution for \setCoverEst.

\begin{lemma}\label{lem:dist-P}
	For any constant $\delta < 1/8$, $\alpha = o(\sqrt{n/\log{n}})$, and $m = \poly(n)$, 	
	\begin{align*}
		\CC{\setCoverEst}{\dister}{\delta} = \Omgt(mn/\alpha^2)
	\end{align*}
\end{lemma}

\begin{proof}
  We prove this lemma using a reduction from \setCoverEst over the distribution
  $\diste$. Let $\protext$ be a $\delta$-error protocol over the distribution $\dister$. Let $\delta' = 3/8 + \delta$; in
  the following, we create a $\delta'$-error protocol $\protsc$ for the distribution
  $\diste$ (using $\protext$ as a subroutine). 
  
  Consider an instance of the problem from the
  distribution $\diste$.  Define a mapping $\sigma: [m+1] \mapsto \SS$ such that for
  $i \leq m$, $\sigma(i) = S_i$ and $\sigma(m+1) = T$.  Alice and Bob use public
  randomness to partition the set of integers $[m+1]$ between each other, assigning each number in
  $[m+1]$ to Alice (resp. to Bob) with probability $1/2$. Note that by
  Remark~\ref{rem:known-i}, we may assume that Bob knows
  the special index $\istar$.
  
  Consider the random partitioning of $[m+1]$ done by the players. If
  $\istar = \sigma^{-1}(S_\istar)$ is assigned to Bob, or $m+1 = \sigma^{-1}(T)$ is
  assigned to Alice, Bob always outputs $2$.  Otherwise, Bob samples one set from $\FC$ for
  each $j$ assigned to him independently and uniformly at random and treat these sets plus
  the set $T$ as his ``new input''.  Moreover, Alice discards the sets $S_j = \sigma(j)$,
  where $j$ is assigned to Bob and similarly treat the remaining set as her new input. The
  players now run the protocol $\protext$ over this distribution and Bob outputs the
  estimate returned by $\protext$ as his estimate of the set cover size.

  Let $\bR$ denote the randomness of the reduction (\emph{excluding} the inner randomness of $\protext$). 
  Define $\event$ as the event that in the described reduction, $\istar$ is assigned to
  Alice and $m+1$ is assigned to Bob. Let $\distnew$ be the distribution of the instances
  over the \emph{new inputs} of Alice and Bob (i.e., the input in which Alice drops the sets assigned to Bob,
  and Bob randomly generates the sets assigned to Alice) when $\event$ happens.  Similarly, we
  define $\hevent$ to be the event that in the distribution $\dister$, $S_\istar$ is
  assigned to Alice and $T$ is assigned to Bob. It is straightforward to verify that
  $\distnew = (\dister \mid \hevent)$. We now have,
	\begin{align*}
		\Pr_{\diste,\bR}\Paren{\protsc \errs} &= \frac{1}{2} \cdot \Pr_{\bR}\Paren{\bar{\event}} + \Pr_{\bR}\Paren{\event} \cdot \Pr_{\distnew}\Paren{\protext \errs} 
		\tag{since Bob outputs $2$ when $\bar{\event}$, he will succeed with probability $1/2$} \\ 
		&=  \frac{1}{2} \cdot \Pr_{\bR}\Paren{\bar{\event}} + \Pr_{\bR}\Paren{\event} \cdot \Pr_{\dister}\Paren{\protext \errs \mid \hevent} 
		 \tag{$\distnew = (\dister \mid \hevent)$} \\
		&\leq \frac{1}{2} \cdot \Pr_{\bR}\Paren{\bar{\event}} + \Pr_{\bR}\Paren{\event} \cdot \frac{\Pr_{\dister}\Paren{\protext \errs}}{\Pr_{\dister}\Paren{\hevent}} \\
		&= \frac{1}{2} \cdot \Pr_{\bR}\Paren{\bar{\event}} + \Pr_{\dister}\Paren{\protext \errs} 	\tag{$\Pr_{\bR}(\event) = \Pr_{\dister}(\hevent)$} \\
		&\leq \frac{3}{8} + \delta \tag{$\Pr_{\bR}(\bar{\event}) = 3/4$} 
	\end{align*}
	Finally, since $\delta < 1/8$, we obtain a $(1/2-\eps)$-error protocol (for some constant $\eps$ bounded away from $0$) for the distribution $\diste$. 
	The lower bound now follows from Theorem~\ref{thm:two-adv}. 
\end{proof}

We can now prove the lower bound for the random arrival model.
\begin{proof}[Proof of Theorem~\ref{thm:rand-lower}]
  Suppose $\alg$ is a randomized single-pass streaming algorithm satisfying the conditions
  in the theorem statement.  We use $\alg$ to create a $\delta$-error protocol \setCoverEst over the distribution $\dister$
   with parameter $\delta = 0.1 < 1/8$.
  
  Consider any input $\SS$ in the distribution $\dister$ and denote the sets given to
  Alice by $\SS_{A}$ and the sets given to Bob by $\SS_{B}$. Alice creates a stream created by a
  random permutation of $\SS_A$ denoted by $s_A$, and Bob does the same for $\SS_B$ and obtains
  $s_B$. The players can now compute $\alg\paren{\seq{s_A,s_B}}$ to estimate the set cover size
  and the communication complexity of this protocol equals the space
  complexity of $\alg$. Moreover, partitioning made in the distribution $\dister$ together
  with the choice of random permutations made by the players, ensures that $\seq{s_A,s_B}$
  is a random permutation of the original set $\SS$. Hence, the probability that $\alg$
  fails to output an $\alpha$-estimate of the set cover problem is at most $\delta = 0.1$.
  The lower bound now follows from Lemma~\ref{lem:dist-P}.
\end{proof}


\section{An $\Omega(mn/\alpha)$-Space Lower Bound for Deterministic $\alpha$-Estimation}\label{sec:det-lower}

In Section~\ref{sec:general}, we provided a \emph{randomized} algorithm for
$\alpha$-estimating the set cover problem, with an (essentially) optimal space bound (as
we proved in Section~\ref{sec:rand-lower}); here, we establish that randomization is
crucial for achieving better space bound for $\alpha$-estimation: any \emph{deterministic}
$\alpha$-estimation algorithm for the set cover problem requires
$\Omega(\frac{mn}{\alpha})$ bits of space. In other words, $\alpha$-approximation and
$\alpha$-estimation are as \emph{hard} as each other for deterministic
algorithms. Formally,

\begin{theorem}\label{thm:det-lower}
  For any $\alpha = o(\sqrt{\frac{n}{\log{n}}})$ and $m = \poly(n)$, any  \emph{deterministic} $\alpha$-estimation single-pass 
  streaming algorithm for the set cover problem requires ${\Omega}(mn/\alpha)$ bits of space.
\end{theorem}

Before continuing, we make the following remark. The previous best lower bound for \emph{deterministic} algorithm asserts that any $O(1)$-estimation algorithm requires $\Omega(mn)$ bits of space even allowing 
\emph{constant number of passes}~\cite{DemaineIMV14}. This lower bound can be stated more generally in terms of its dependence on $\alpha$: for any $\alpha \leq \log{n}-O(\log\log{n})$, any deterministic $\alpha$-estimation
algorithm requires space of $\Omega(mn/\alpha\cdot2^{\alpha})$ bits (see Lemma~11 and Theorem~3 in~\cite{DemaineIMV14}). Our bound in Theorem~\ref{thm:det-lower}
provide an exponential improvement over this bound for \emph{single-pass} algorithms.

Recall the definition of \setCoverEst from Section~\ref{sec:rand-lower} for a fixed choice of parameters $n,m$, and $\alpha$. We define the following hard input distribution for deterministic protocols of \setCoverEst. 

\cTextbox{Distribution $\distdet$. \textnormal{A hard input distribution for \emph{deterministic} protocols of \setCoverEst.}}{ 
	\medskip \\
	\textbf{Notation.} Let $\FC$ be the collection of all subsets of $[n]$ with cardinality $n/10\alpha$. 
\begin{itemize}
	\item \textbf{Alice.} The input to Alice is an $m$-subset $\SA \subseteq \FC$ chosen uniformly at random. 
	 
	\item \textbf{Bob.} Toss a coin $\theta \in_R \set{0,1}$; let the input to Bob be the set $T$, where if $\theta = 0$, $\bar{T} \in_R \SA$ and otherwise, $\bar{T} \in_R \FC \setminus \SA$.
\end{itemize}
}

We stress that distribution $\distdet$ can \emph{only} be hard for deterministic protocols;  Alice can simply run a \emph{randomized} protocol for checking the equality of each $S_i$ with $\bar{T}$ (independently for 
each $S_i \in \SA$) and since the one-way communication complexity of the \emph{Equality} problem with public randomness is $O(1)$ (see, e.g.,~\cite{KN97}), the total communication complexity of this protocol is $O(m)$. 
However, such approach is not possible for deterministic protocols since deterministic communication complexity of the \emph{Equality} problem is linear in the input size (again, see~\cite{KN97}). 

The following lemma can be proven similar to Lemma~\ref{lem:theta-random} (assuming that $\alpha = o(\sqrt{\frac{n}{\log{n}}})$). 
\begin{lemma}\label{lem:theta-det}
	For $(\SA,T) \sim \distdet$: 
	\begin{enumerate}[(i)]
		\item \label{part:det-theta0} $\PR{\opt(\SA,T) = 2 \mid \theta=0} = 1$. 
		\item \label{part:det-theta1} $\PR{\opt(\SA,T) > 2\alpha \mid \theta = 1} = 1-2^{-\Omega(n/\alpha)}$. 
	\end{enumerate}
\end{lemma}


\newcommand{\sparseIndexing}{\emph{Sparse Indexing}\xspace}

To prove Theorem~\ref{thm:det-lower}, we use a reduction from the \emph{Sparse Indexing} problem of~\cite{SaglamThesis2011}. 
In \sparseIndexingnk, Alice is given a set $S$ of $k$ elements from a universe $[N]$ and Bob is given an 
element $e \in [N]$; Bob needs to output whether or not $e \in S$. The crucial property of this problem that we need in our proof is that
the communication complexity of \sparseIndexing depends on the success probability required from the protocol. 

The following lemma is a restatement of Theorem 3.2 from~\cite{SaglamThesis2011}.

\begin{lemma}[\!\!\cite{SaglamThesis2011}] \label{lem:sparse-index}
	Consider the following distribution $\distsi$ for \sparseIndexingnk: 
	\begin{itemize}
	\item Alice is given a $k$-subset $S \subseteq [N]$ chosen uniformly at random. 
	\item With probability half Bob is given an element $e \in S$ uniformly at random, and with probability half Bob receives a random element $e \in [N] \setminus S$.
	\end{itemize} 
	For any $\delta < 1/3$, the communication complexity of any deterministic $\delta$-error protocol over distribution $\distsi$ is $\Omega(\min\set{k\log(1/\delta),k\log(N/k)})$ bits.
\end{lemma}

We now show that how a deterministic protocol for \setCoverEst can be used to solve the
hard distribution of \sparseIndexingnk described in the previous lemma.

Notice that there is a simple one-to-one mapping $\phi$ from an instance $(S,e)$ of
\sparseIndexingnk to an instance $(\SA,T)$ of \setCoverEst with parameters $n$ and $m$
where $N = {n \choose \frac{n}{10\alpha}}$ and $k = m$. Fix any arbitrary bijection
$\sigma: [N] \mapsto \FC$; we map $S$ to the collection
$\SA = \set{\sigma(j) \mid j \in S}$ and map $e$ to the set $T = [n] \setminus \sigma(e)$.
Under this transformation, if we choose $(S,e) \sim \distsi$ the \setCoverEst
distribution induced by $\phi(S,e)$ is equivalent to the distribution
$\distdet$. Moreover, the answer of $(S,e)$ for the \sparseIndexingnk problem is \Yes iff
$\theta = 0$ in $(\SA,T)$. We now proceed to the proof of Theorem~\ref{thm:det-lower}.

\begin{proof}[Proof of Theorem~\ref{thm:det-lower}] 
	Let $\protsc$ be a deterministic protocol for the \setCoverEst problem; we use $\protsc$ to design a $\delta$-error protocol $\protsi$ for \sparseIndexingnk on the distribution $\distsi$ for $\delta = 2^{-\Omega(n/\alpha)}$.
	The lower bound on message size of $\protsc$ then follows from Lemma~\ref{lem:sparse-index}. 
	
	Protocol $\protsi$ works as follows. Consider the mapping $\phi$ from $(S,e)$ to $(\SA,T)$ defined before. Given an instance $(S,e) \sim \distsi$, Alice and Bob can each
	 compute their input in $\phi(S,e)$ without any communication. Next, they run the protocol $\protsc(\SA,T)$ and if the set cover size is less than $2\alpha$ Bob outputs \Yes and otherwise outputs \No.
	
	We argue that $\protsi$ is a $\delta$-error protocol for \sparseIndexingnk on the distribution $\distsi$.  
	To see this, recall that if $(S,e) \sim \distsi$ then $(\SA,T) \sim \distdet$. Let $\event$ be an event in the distribution \distdet, where $\theta = 1$ \emph{but} $\opt < 2\alpha$. Note that since $\protsc$ is a 
	deterministic $\alpha$-approximation protocol and never errs, the answer for $\protsi$ would be wrong only if $\event$ happens. 
	We have,  
	\begin{align*}
		\Pr_{(S,e) \sim \distsi}\Paren{\protsi(S,e) \errs} &\leq \Pr_{(\SA,T) \sim \distdet}\Paren{\event} = 2^{-\Omega(n/\alpha)} 
	\end{align*}
	where the equality is by Lemma~\ref{lem:theta-det}. This implies that $\protsi$ is a $\delta$-error protocol for the distribution $\distsi$. 
	By Lemma~\ref{lem:sparse-index}, 
	\[
		\norm{\protsc} = \norm{\protsi} = \Omega(\min(k\cdot \log{1/\delta}, k\log(N/k))) = \Omega(mn/\alpha)
	\]
	As one-way communication complexity is a lower bound on the space complexity of single-pass streaming algorithm we obtain the final result.  
\end{proof}

\subsection*{Acknowledgements}
We thank Piotr Indyk for introducing us to the problem of determining single-pass streaming complexity of set cover,
and the organizers of the DIMACS Workshop on Big Data through the Lens of Sublinear Algorithms (August 2015) where this conversation happened. We are also 
thankful to anonymous reviewers for many valuable comments. 

\clearpage
\bibliographystyle{acm}
\bibliography{general}

\clearpage
\appendix
\newcommand{\ba}{\ensuremath{\bm{a}}}
\section{A Lower Bound for $\IC{\Indexnk}{\distI}{\delta}$}\label{app:index}

In this section we prove Lemma~\ref{lem:index}. As stated earlier, this lemma is well-known and is proved here only for the sake of completeness. 

\begin{proof}[Proof of Lemma~\ref{lem:index}]
	Consider the following ``modification'' to the distribution $\distI$. Suppose we first sample a $2k$-set $P \subseteq [n]$, and then let the input to Alice be a $k$-set $A \subset P$ chosen uniformly at random, and 
	the input to Bob be an element $a \in P$ chosen uniformly at random. It is an easy exercise to verify that this modification does not change the distribution $\distI$.  
	However, this allows us to analyze the information cost of any protocol $\protindex$ over $\distI$ easier. In particular, 
	\begin{align*}
		\ICost{\protindex}{\distI} &= I(\bA ; \bprotindex) \\
		&\geq I(\bA; \bprotindex \mid \bP) \\
		&= H(\bA \mid \bP) - H(\bA \mid \bprotindex, \bP) \\
		&= 2k - \Theta(\log{k}) - H(\bA \mid \bprotindex, \bP)
	\end{align*}
	where the inequality holds since $(i)$ $H(\bprotindex) \geq H(\bprotindex \mid \bP)$ and $(ii)$ $H(\bprotindex \mid \bA) = H(\bprotindex \mid \bA,\bP)$ as $\bprotindex$ is independent of $\bP$ conditioned on $\bA$.
	 
	We now bound $H(\bA \mid \bprotindex, \bP)$. Define $\theta \in \set{0,1}$ where $\theta = 1$ iff $a \in A$. Note that a $\delta'$-error protocol for $\Index$ is also 
	a $\delta'$-error (randomized) estimator for $\theta$ (given the message $\protindex$, the public randomness used along with the message, and the element $a$). Hence, using Fano's inequality (Claim~\ref{clm:fano}),  
	\begin{align*}
		H_2(\delta') &\geq H(\bT \mid \bprotindex,\ba) \\
		&\geq H(\bT \mid \bprotindex, \ba, \bP) \tag{conditioning decreases entropy, Claim~\ref{clm:it-facts}-(\ref{part:cond-reduce})} \\
		&= \Ex_{P \sim \bP} \Ex_{a \sim \ba \mid \bP = P} \Bracket{H(\bT \mid \bprotindex, \ba =a , \bP = P)} 
	\end{align*}
	Conditioned on $\bP = P$, $\ba$ is chosen uniformly at random from $P$. Define $\bX:=(\bX_1,\ldots,\bX_{2k})$, where 
	$\bX_i = 1$ if $i$-th element in $P$ belongs to $\bA$ and $\bX_i = 0$ otherwise. Using this notation, we have $\bT = \bX_i$ conditioned on $\bP = P$, and $\ba = a$ being the $i$-th element of $\bP$. Hence,
	\begin{align*}
	H_2(\delta') &\geq \Ex_{P \sim \bP} \Bracket{\sum_{i=1}^{2k} \frac{1}{2k} \cdot H(\bX_i \mid \bprotindex, \ba = a, \bP = P)} \\
		&= \Ex_{P \sim \bP} \Bracket{\sum_{i=1}^{2k} \frac{1}{2k} \cdot H(\bX_i \mid \bprotindex, \bP = P)} \tag{$\bX_i$ and $\bprotindex$ are independent of $\ba = a$ \emph{conditioned
		 on $\bP = P$}, and Claim~\ref{clm:it-facts}-(\ref{part:ent-event})}\\
		&\geq \frac{1}{2k} \cdot \Ex_{P \sim \bP} \Bracket{ H(\bA \mid \bprotindex,  \bP = P)} \tag{sub-additivity of the entropy, Claim~\ref{clm:it-facts}-(\ref{part:sub-additivity})} \\
		&= \frac{1}{2k} H(\bA \mid \bprotindex,\bP)
	\end{align*}
	Consequently, 
	\begin{align*}
		H(\bA \mid \bprotindex,\bP) \leq 2k \cdot H_2(\delta')
	\end{align*}
	Finally, since $\delta' < 1/2$ and hence $H_2(\delta') = (1-\eps)$ for some $\eps$ bounded away from $0$, we have, 
	\begin{align*}
		\ICost{\protindex}{\distI} &\geq 2k \cdot (1-H_2(\delta')) - \Theta(\log{k}) = \Omega(k)	
	\end{align*}
	
\end{proof}

\end{document}